\def\isarxiv{1} 

\ifdefined\isarxiv
\documentclass[11pt]{article}

\else
\documentclass{article}
\usepackage{hyperref}

\fi
\usepackage[numbers]{natbib}
\usepackage{amsmath}
\usepackage{amsthm}
\usepackage{amssymb}
\usepackage{algorithm}
\usepackage{algpseudocode}
\usepackage{graphicx}
\usepackage{grffile}
\usepackage{wrapfig,epsfig}
\usepackage{url}
\usepackage{xcolor}
\usepackage{epstopdf}
\usepackage{accents}
\usepackage{enumerate}
\usepackage{nicefrac}
\usepackage{subfigure}

\usepackage{tikz}
\usetikzlibrary{arrows,positioning}

\usepackage{booktabs}

\usepackage{amssymb}
\usepackage{optidef}

\usepackage{bbm}
\usepackage{dsfont}

\usepackage{tikz}

\allowdisplaybreaks

\ifdefined\isarxiv

\usepackage{tikz}
\usepackage{hyperref}  
\hypersetup{colorlinks=true,citecolor=blue,linkcolor=blue} 
\usetikzlibrary{arrows}
\usepackage[margin=1in]{geometry}

\else

\usepackage{microtype}

\fi
\graphicspath{{./figs/}}
 
\newcommand{\squishlist}{
  \begin{list}{$\bullet$}{
    \setlength{\itemsep}{0pt}       \setlength{\parsep}{3pt}
    \setlength{\topsep}{3pt}        \setlength{\partopsep}{0pt}
    \setlength{\leftmargin}{1em}    \setlength{\labelwidth}{1em}
    \setlength{\labelsep}{0.5em} } }
\newcommand{\squishend}{
  \end{list} }

\newtheorem{theorem}{Theorem}[section]

\newtheorem{lemma}[theorem]{Lemma}
\newtheorem{definition}[theorem]{Definition}

\newtheorem{assumption}[theorem]{Assumption}

\newtheorem{problem}[theorem]{Problem}

\newcommand{\wh}{\widehat}
\newcommand{\wt}{\widetilde}

\newcommand{\R}{\mathbb{R}}

\renewcommand{\tilde}{\wt}
\renewcommand{\hat}{\wh}

\DeclareMathOperator*{\E}{{\mathbb{E}}}

\DeclareMathOperator*{\D}{\mathcal{D}}

\DeclareMathOperator{\OPT}{OPT}

\DeclareMathOperator{\poly}{poly}

\newcommand{\V}{\mathcal{V}}

\newcommand{\ve}{\mathbf{e}}
\newcommand{\vp}{\mathbf{p}}

\newcommand{\B}{\mathcal{B}} 
\newcommand{\hyperb}{\mathcal{L}}

\newcommand{\frm}{R_\textsc{FRM}}
\newcommand{\frc}{c_\text{FR}}

\newcommand{\util}{\textsc{Util}}
\newcommand{\tmix}{t_\text{mix}}

\newcommand{\dtv}{d_\text{TV}}
\newcommand{\trans}{{P_M}} 
\newcommand{\eps}{\epsilon}

\newcommand{\T}{\mathcal{T}}

\definecolor{myblue}{RGB}{51,153,255}

\definecolor{myred}{RGB}{200,5,5}
\definecolor{ForestGreen}{RGB}{34,139,34}

\makeatletter

\usepackage[textsize=tiny,textwidth=0.6in]{todonotes}
\newcommand{\allnotes}[1]{}
\renewcommand{\allnotes}[1]{\textit{#1}}

\newcommand{\ruimin}[1]{\allnotes{\todo[color=ForestGreen]{Ruimin: #1}}}

\usepackage{lineno}

\renewcommand{\citet}{\cite} 

\begin{document}

\ifdefined\isarxiv

 \title{Markovian Pandora's box}

\author{
Yuanyuan Yang \thanks{University of Washington. Email:\texttt{yyangh@cs.washington.edu}. }
\and 
Ruimin Zhang \thanks{University of Chicago. Email: \texttt{ruimin@uchicago.edu}.}
\and 
Jamie Morgenstern \thanks{University of Washington. Email:\texttt{jamiemmt@cs.washington.edu}.}
\and
Haifeng Xu \thanks{University of Chicago. Email: \texttt{haifengxu@uchicago.edu}.}
}

\else

\fi

\ifdefined\isarxiv
\begin{titlepage}
  \maketitle
  \begin{abstract}

In this paper, we study the \emph{Markovian} Pandora’s Box Problem, where decisions are governed by both order constraints and Markovianly correlated rewards, structured within a shared directed acyclic graph (DAG). To the best of our knowledge, previous work has not incorporated \emph{Markovian dependencies} in this setting. This framework is particularly relevant to applications such as data or computation driven algorithm design, where exploration of future models incurs cost. 

We present optimal \emph{fully adaptive} strategies where the associated graph forms a forest. Under static transition, we introduce a strategy that achieves a near-optimal expected payoff in multi-line graphs and a $1/2$-approximation in forest-structured graphs. Notably, this algorithm provides a significant speedup over the exact solution, with the improvement becoming more pronounced as the graph size increases. Our findings deepen the understanding of sequential exploration under Markovian correlations in graph-based decision-making.

  \end{abstract}
  \thispagestyle{empty}
\end{titlepage}

\else

\begin{abstract}

In this paper, we study the \emph{Markovian} Pandora’s Box Problem, where decisions are governed by both order constraints and Markovianly correlated rewards, structured within a shared directed acyclic graph (DAG). To the best of our knowledge, previous work has not incorporated \emph{Markovian dependencies} in this setting. This framework is particularly relevant to applications such as data or computation driven algorithm design, where exploration of future models incurs cost. 

We present optimal \emph{fully adaptive} strategies where the associated graph forms a forest. Under static transition, we introduce a strategy that achieves a near-optimal expected payoff in multi-line graphs and a $1/2$-approximation in forest-structured graphs. Notably, this algorithm provides a significant speedup over the exact solution, with the improvement becoming more pronounced as the graph size increases. Our findings deepen the understanding of sequential exploration under Markovian correlations in graph-based decision-making.

\end{abstract}

\fi

\section{Introduction}

Uncertainty is a fundamental challenge in decision-making for optimization problems. It can often be mitigated through costly inspections, where incurring a cost reveals additional information, as seen in delegated search~\cite{kk18}, ranking~\cite{dgmm22}, and hyperparameter tuning~\cite{smv+20}, etc. A foundational framework for decision-making under costly information is the Pandora’s Box problem, originally formulated by~\cite{w79}. In this model, a decision maker is presented with a collection of boxes, each associated with a known fixed cost and an unknown reward drawn from a known distribution. The objective is to determine an optimal (adaptive) order, where the decision maker can choose to stop at any point, selecting the highest observed reward while incurring the cumulative cost of the opened boxes. The goal is to maximize the expected net payoff. This problem has been extensively studied and extensively generalized across multiple disciplines, including correlated reward distributions~\cite{cgt+20}, box-dependent deadlines~\cite{beff24}, order constraints~\cite{bfll20}, and online settings~\cite{gt22}, etc.

Existing literature fails to capture two critical features inherent in many practical applications: 1)\ the presence of \emph{Markov dependencies} among the boxes, where the reward of certain boxes is determined by others, but not vice versa, and 2)\ an \emph{order constraint} associated with this Markov dependency. Consider data-driven algorithm design~\cite{gr16} as an example: a simpler model (e.g., a checkpoint after 200 gradient descent steps) must be computed before a more refined model (e.g., a fully converged model with extra computation ) becomes accessible. This dependency structure induces a Markovian relationship, where the reward of a later box is conditioned on an earlier one, while simultaneously imposing an order constraint on the exploration process.


\textbf{Markovian Pandora's Box}. Building on the real-world motivations discussed above, we introduce and analyze the \emph{Markovian Pandora’s Problem}, an extension of the classic Pandora’s Box framework that incorporates both \emph{structural constraints} and \emph{probabilistic dependencies}. Specifically, we augment the standard problem with a \emph{directed acyclic graph (DAG)}, which simultaneously encodes \emph{precedence constraints} and \emph{Markovian correlations} among boxes. More precisely, a directed edge from box $A$ to box $B$ implies that: 1)\ $A$ must be probed before $B$, and 2)\ the reward of $B$ depends on the revealed reward of $A$ Markovianly.


\textbf{Curse of Adaptivity}. The optimal strategy of Markovian Pandora's Box is \emph{fully adaptive}(FA), i.e., the order of the strategy changes adaptive to the all the realized reward. Two challenges : 1)\ It's hard to characterize an \emph{explicit form} of the fully adaptive strategy, as it is typically computed via exhaustive search methods, and 2)\ Optimizing an FA strategy involves solving a high-dimensional adaptive decision process, which is inherently \emph{NP-hard}~\cite{cgmt21}.

\subsection{Results}\label{subsec:results_intro}

\subsubsection{Exact Optimization: Structured Solutions via Equivalence Reward}
In this paper, we present the \emph{first} optimal algorithm for the Markovian Pandora’s Box problem on a forest-structured graph. Despite its fully adaptive nature, we show that the optimal solution can be computed efficiently in polynomial time and space. We start with the simplest case of a \emph{single-line} precedence graph, gradually extending the strategy to \emph{multiple-line} settings, and to the \emph{forest}-structured case.

For the \emph{single-line} case (Section~\ref{sec:1_hyperbox}), we show that the optimal strategy simplifies to a \emph{stopping time}, as the probing order remains \emph{fixed}. Our key idea is to construct a \emph{polynomial-sized equivalent reward table}, computable in polynomial time, which determines the optimal stopping time.  

We then generalize the above setting to the \emph{multi-line} setting (Section~\ref{sec:multi_boxes}), where now the optimal strategy includes deciding an \emph{adaptive} order. We prove that, even in this setting, the optimal strategy remains governed by the equivalent reward table, which is unaffected by order adjustments based on reward realizations. Consequently, we derive an optimal strategy in polynomial time and space.

Finally, this framework extends to the forest setting (Section~\ref{sec:multi_boxes}) through \emph{contraction}, where the equivalent reward table for a multi-line structure is transformed into that of a single box with a random cost. By iteratively applying this contraction, the precedence graph is systematically reduced to a multi-line structure. Throughout this process, the payoff table updates with each newly opened box, while maintaining polynomial time and space complexity:

\begin{theorem}[Optimal Solution for Forest-Structured Graphs (Lem.~\ref{lem:app:update_GRV_forest} and Thm.~\ref{thm:app:GRV_opt_forest})]\label{thm:MPB_forest}
    For a Markovian Pandora’s Box problem with a forest-structured precedence graph, there exists a fully adaptive algorithm that achieves the optimal expected payoff in polynomial time and space.
\end{theorem}

\subsubsection{Static Transition: Faster Solution Via Subgraph Selection}

In the context of \emph{static transition}, we develop \emph{faster} strategies for single-line, multi-line, and forest precedence graphs, addressing each case separately. Here, static transition refers to a setting where the value transition—defined as the conditional reward distribution of the box at the end of a directed edge given the box at the start—remains identical across all edges within the same component. 

Our first result explores a special case with multiple infinite-length lines where each box has the same cost. In this setting, the equivalence table is no longer correlated with the index of the current box. We solve for this table using fixed-point iteration. See Thm.~\ref{thm:MC_constant_cost} for more details.

Our second result develops a faster algorithm for probing the Markovian Pandora’s Box problem under line (Thm.~\ref{thm:greedy_optimal_single_line}), multi-line (Thm.~\ref{thm:greedy_optimal_multi_line}), and forest constraints (Thm.~\ref{thm:forest_approx}). Our solution leverages subgraph optimization, where for a given $\delta \in (0,1)$, we identify a subgraph depends only on polylog factors of $\delta$. Exploring this subgraph is near-optimal up to an additive $\Theta(\delta)$ in the line and multi-line cases, while in the forest setting, it achieves (roughly) a $1/2$-approximation (Def.~\ref{def:approx_ratio}):

\begin{theorem}[Faster Solution Under Static Transition]\label{thm:sum_static_transition}
Given a Markovian Pandora’s box with static transition,
\squishlist
    \item \textbf{Multi Lines}. We can find a subgraph $\hat{G}$ of size $\Tilde{\Theta} (q)$ in $\Tilde{\Theta} (q)$ time, such that probing against $\hat{G}$ is optimal up to $\Theta(q \delta)$, where $q$ is the number of distinct lines. 
    \item \textbf{Forest}. We can find a subgraph $\hat{G}$ of size at most $\Tilde{\Theta} (1)$ in $\Delta(G)^{\tilde{\Theta}(1)}$ time, such that probing against $\hat{G}$ satisfies:
    {\small 
    \begin{align*}
        \E [\max_{i \in \hat{G}} R_i - &\sum_{i \in \hat{G}} c_i] \geq 1/2 \cdot  \E[\max_{i \in \mathcal{O}(\hat{\pi})} R_i] - \sum_{i \in \mathcal{O}(\hat{\pi})} c_i] -q \delta.
    \end{align*}
    }
    against the boxes $\mathcal{O}(\hat{\pi})$ selected by any strategy $\hat{\pi}$. 
    
\squishend
    
\end{theorem}



\subsection{Literature Review}

\textbf{Pandora's Box Problem}.
This problem originates from \cite{w79}. We recommend ~\cite{bc24} for recent developments on this problem. The most related work is that of order constraints and with correlations:

Previous work by {bfll20} focused on order constraints, where some boxes must be opened after others, and rewards are independent accross boxes, making it a \emph{special case} of ours. Their optimal strategy is partially adaptive, whereas ours is fully adaptive. Consequently, our setting is fundamentally more challenging than theirs.

Recent studies have shown growing interest in the Pandora’s Box problem with correlations~\cite{cdkt19, cgt+20, cgmt21, gt23}. These works study the cost minimization version of the Pandora’s Box problem and focus on deriving adaptive strategies that approximate the fully adaptive (FA) or partially adaptive optimal solutions. Their results show that approximating the FA optimal within a constant factor is \emph{NP-hard}. In contrast, our setting assumes structured correlations, allowing for the exact optimization of FA strategies within polynomial time. 



\textbf{Data Driven Algorithm Design}. Our framework relates to data-driven algorithm design~\cite{gr16} with cost, where practitioners refine parameterized algorithms via training instances to maximize expected future performance, including~\cite{ ggm06, bb12, jt16, ldrt17, hky18}.

 We defer to Appendix.~\ref{sec:app:literature_review} for more detailed literature review.


\section{Problem Formulation}\label{sec: problem_formulation}

\subsection{Markovian's Pandora's Box}\label{subsec:pandora_pre}

In this paper, we consider the Markovian Pandora's box with order constraints. Subject to (partial) ordering constraints, some boxes must be opened after others. These boxes have \emph{known fixed} probing costs, and payoffs correlated across boxes in a \emph{Markovian} fashion given by the underlying directed graph of the order constraints.

\begin{problem}[Markovian Pandora's Box]\label{prob:mar_pandora}
    Given a set of $n$ boxes $\B= \{b_1, \cdots, b_n\}$. For every $i \in [n]$, box $b_i$ has a \emph{known} fixed probing cost $c_i$, and a random reward $R_i$, where $R_i$ follows a \emph{known} distribution $\D_i$. The rewards of the boxes are correlated in a Markov fashion. The correlation of the rewards and the order constraints of probing the boxes are given by the same \emph{directed} acyclic graph $G = (\B, E)$, where the boxes are vertices connected by directed edges in $E$. 
    \squishlist
    \item \textbf{Partial Ordering}: For any edge $(b_i, b_j) \in E$, box $b_i$ must be probed before box $b_j$, and we use $b_i \prec b_j$ to denote this relation.
    \item \textbf{Markov property}: For any boxes $b_i \prec  b_j \prec b_k$ that forms a directed line in $G$, then the reward of $b_i$ and $b_k$ is conditionally independent given the reward of $b_j$: 
    \begin{align*}
        &\Pr [R_i =y, R_k = x | R_j = z] = \\
        & \quad \quad \Pr [R_i =y | R_j  = z] \cdot \Pr[R_k = x | R_j = z]
    \end{align*}
    for any $x,y,z \in \R_{+}$.
    \squishend

Our goal is to find a policy $\pi^*$ that iteratively probes the boxes, and maximizes the expected payoff, defined as the expected maximum reward minus the total probing costs: 
\begin{equation}
    \E \Big[\max_{i \in \mathcal{O}(\pi^*)} R_i - \sum_{i\in \mathcal{O}(\pi^*)}c_i\Big]
\end{equation}
where $\mathcal{O} (\pi)$ denote the (random) set of boxes opened following strategy(policy) $\pi$. 
\end{problem}

\subsection{Adaptivity Gap}\label{subsec:adapt_gap}
We introduce three classes of strategies which have different level of adaptivities. 

\begin{definition}[Adaptivity in Strategy Design: NA, PA, FA]\label{def:adaptive_strategy}
The strategies in Markovian Pandora's box are defined by an ordering $\omega$ and a stopping time $\tau$. A strategy is:
\squishlist
    \item \textbf{Non-adaptive}: if both $\tau$ and $\omega$ are independent of the realized rewards.
    \item \textbf{Partially adaptive}: if $\omega$ is independent of the realized rewards, but $\tau$ depends on them.
    \item \textbf{Fully adaptive}: if both $\tau$ and $\omega$ depend on the realized rewards.
\squishend
\end{definition}


The optimal strategy for the classic Pandora’s Box Problem, known as Weitzman’s rule \cite{w79}, is \emph{partially adaptive}. It assigns each box a reservation value based on its reward distribution and cost and probes boxes in decreasing order of these values. The process stops once a sufficiently high reward is found. For Pandora’s Box with correlation, the optimal strategy becomes \emph{fully adaptive}. While Weitzman’s rule provides a constant-factor approximation to the best PA strategy \cite{gt23}, achieving a constant-factor approximation against the best FA strategy is NP-hard \cite{cgmt21}. This makes it valuable to develop strategies that approximate the best PA strategy, as explored in \cite{cgt+20}. Similarly, in Pandora’s Box with order constraints \cite{bfll20}, PA strategies remain a key focus of study, as they achieve a constant-factor approximation of the best FA strategy while preserving computational efficiency.

Our solution represents the optimal strategy against \emph{fully adaptive} strategies, and we derive a closed-form solution when the underlying DAG forms a forest. Given that FA strategies are the most complex class of strategies to analyze, one might wonder if there are adaptivity gap between the performance against the best FA strategies and others. The following lemma demonstrates that PA strategies are not optimal, justifying the need to consider FA strategies.

\begin{lemma}[The sub-optimality of PA strategies]\label{lem:sub_opt_PA}
There exist an instance of Markovian Pandora's box (Lem.~\ref{lem:app:sub_opt_PA}) where the best FA strategy outperforms best PA strategies. 
\end{lemma}

Since \cite{bfll20} examine a special case of our Markovian Pandora’s Box model, their negative result also implies that finding the optimal solution under a general DAG is \emph{NP-hard} in our setting. To maintain tractability, we focus on precedence graphs that form a \emph{forest}.

\begin{theorem}[Lower bound for Pandora's problem]
Computing a $0.9997$-approximate optimal solution for the Pandora’s Box Problem with order constraints, where the precedence graph forms a DAG, is \emph{NP-hard}.
\end{theorem}

\textbf{Notations}.
For the remainder of the paper, we focus on the case where reward distributions have \emph{finite support}. WLOG, we let $n$ denote the number of boxes and assume that all boxes share the same finite set of possible values. Specifically, each reward $R_i$ for $i \in [n]$ takes $k$ values from $V := \{v_1, \ldots, v_k\}$. For every $i \in [n]$, we use $s^i$ and $R_i$ interchangeably to denote the reward of box $b_i$, and we denote the probability density function (pdf) of its reward as $\vp_i$, i.e., $\vp_i[s_q] = \Pr [R_i = s_q]$. We use $P_i \in \R_{+}^{K \times K}$ to denote the (probability) transition matrix from $\D_i$ to $\D_{i+1}$, i.e., $\vp_{i+1} = \vp_i \cdot P_{i+1}$.

We say an algorithm runs in polynomial time if its running time is in $\text{poly}(k, n)$. We use $\Tilde{O}$ as a variant of the Big-O that ignores polylog factors.

\section{Exact Optimization for Single Line}\label{sec:1_hyperbox}

In this section, we introduce the optimization approach for the Markovian Pandora’s Box with \emph{line} constraints. This solution serves as a fundamental building block for our approach to the subsequent forest and multi-line settings.

\subsection{Algorithm for Line Constraint}\label{subsec:alg_1_hyperbox}

Before presenting the details, we formally define a hyperbox as a sequentially ordered set of boxes forming a directed path within the original graph.

\begin{definition}[Hyperbox]\label{def:hyper}
    Given a instance of Markovian Pandora's box with $n$ boxes $\B$ associated with a DAG $G = (\B, E)$, a hyperbox $\hyperb:= \{b_1, \ldots, b_n\}\subseteq G$ is a subgraph of $G$ such that $\mathcal{L}$ is a directed line.
\end{definition}

Given a Markovian Pandora's box with a line constraint, we are ready to present our optimal solution for the one-line case (Alg.~\ref{alg:pd_1path}), which depends on the generalized reservation value (GRV) (Def.~\ref{def:GRV}) of a hyperbox. More specifically, our solution iteratively evaluate the GRV of the next box when exploring along a line. The process continues until the GRV of the next box falls below the current maximum reward, at which point the search terminates. 

\begin{algorithm}[!ht]\caption{Markov Pandora's Box, Single Line}\label{alg:pd_1path}
\begin{algorithmic}[1]
    \Require Ordered set of boxes $\{b_1,\ldots, b_n\}$, probing cost $\{c_1,\ldots, c_n\}$, GRV $\sigma_i(s)$ for all $i$, $x$ and $s$ (Alg.~\ref{alg:app:gw_1path}).
    \State Initialize $x \gets 0$, $i \gets 1$, $j \gets 1$, $\sigma \gets \sigma_i(0,0)$
    \While{$x < \sigma$} 
        \State Pay $c_i$ to open box $b_i$, observe reward/state $s^i$.
        \State $x \gets \max\{x, s^i\}$. \Comment{Update max reward}
        \State $\sigma \gets \sigma_{i+1}(s^i)$. \Comment{Lookup the GRV}
        \State $i \gets i+1$.
    \EndWhile
    \State \textbf{Return} box $b_j$ that is opened and with the max reward.
\end{algorithmic} 
\end{algorithm}

\subsection{Generalized Reservation Value}\label{subsec:GRV}

In this section, we show that probing according to the GRV maximizes the expected payoff. In the line case, all adaptive strategies are partially adaptive, as the order constraint uniquely determines the probing sequence. Consequently, optimizing the strategy simplifies to finding the optimal adaptive stopping time. 

The optimal stopping time depends on the current state $(x, s^{i-1}, i)$, where $x$ is the highest observed reward, $s^{i-1}$ is the state of the last opened box $b_{i-1}$, and $i$ is the next box available for probing. For any state $(x, s^{i-1}, i)$ and strategy $\tau$, we denote $\tau(x, s^{i-1}, i)$ as the random stopping time conditioned on this state. We derive the equivalent reward given any stopping time $\tau$:


\begin{definition}[Equivalent Reward]\label{def:exp_reward}
    Given $\tau$ and $(x, s^{i-1}, i)$, we define the expected future reward following $\tau$, starting at state $(x, s^{i-1}, i)$, as:
    \[\Phi^{\tau}(x, s^{i-1}, i) := \E [\max \{x , \max _{j=i} ^{\tau(x, s^{i-1}, i)}R_j\} - \sum_{j=i}^{\tau(x, s^{i-1}, i)} c_j]\]
\noindent In addition, we use 
    \[
    \Phi (x, s^{i-1}, i)  = \Phi ^{\tau^*}(x, s^{i-1}, i)= \max _\tau \Phi^{\tau}(x, s^{i-1}, i) 
    \]
to denote the expected future reward following the optimal strategy $\tau^*$ starting at state $(x, s^{i-1}, i)$.
\end{definition}

 \noindent Given an optimal stopping time $\tau^*$, $\Phi$ could be solved inductively by Bellman’s principle of optimality:
  \begin{align*}
     & \Phi^{\tau^*}(x, s^{i-1},i) = \max \{x, -c_i + \\
     & \quad \E[\max \{ \max\{x, s^i\}, \max_{j=i+1}^{\tau^*(x,s^i,i+1)} s^j - \sum_{j=i+1}^{\tau^*(x,s^i,i+1)} c_j \}] \} \\
     & = \max \{x, -c_i + \E_{s^i}[\phi^{\tau^*} (\max\{x,s^i\}, s^i, i+1) ]\}.
 \end{align*}

Inside the max operator, the first term represents the utility of not exploring box $i$, while the second captures the expected utility of optimally exploring future boxes. At a certain threshold $x$, the decision-maker is indifferent between continuing and stopping. We define this threshold as the \emph{generalized reservation value} (GRV):

\begin{definition}[Generalized Reservation Value]\label{def:GRV}
	\label{def:grv}
	Given any state $(x, s^{i-1}, i)$, we define the generalized reservation value at current state for box $b_i$, denoted $\sigma_i$, as the smallest solution to: 
	\begin{equation}
	\E \Big[\Big( \max _{j=i} ^{\tau^*(\sigma_i, s^{i-1}, i)}R_j - \sigma_i\Big)_+ - \sum_{j=i}^{\tau^*(\sigma_i, s^{i-1}, i)} c_j\Big] = 0
	\label{eq:grve}
	\end{equation}
    where $\tau^*$ is the optimal strategy.
\end{definition}



We concluded this section by presenting a few properties of the GRV. More details in Appendix.~\ref{sec:app:multi_line}
\begin{lemma}[Properties of GRV]\label{lem:properties_GRV}
 Given a Markovian Pandora's box with line precedence graph $\mathcal{L}=\left[b_1, \ldots, b_n\right]$, the generalized reservation value of every box $i \in [n]$ satisfies the following property: Given any state $s^{i-1}$,
 \squishlist
    \item $\sigma_i$ is independent of the current max reward $x$.
    \item $\sigma_i(s^{i-1}, i)$ is nondecreasing as additional boxes are appended to $\hyperb$.
    \item Let $\eta$ be the (random) index of the first box that has generalized reservation value smaller than $\sigma_i(s^{i-1}, i)$, then $\sigma_i(s^{i-1}, i)$ depends only on the (sub)hyperbox $\hat{\hyperb} := \{b_i, \ldots, b_{\eta}\}$. If $i = \eta$ with probability $1$, then $\sigma_i(s^{i-1}, i)$ depends only on $b_i$. 
 \squishend

\end{lemma}

\subsection{Correctness and Running Time Analysis}\label{subsec:correct_time_line}

\noindent We first show the existence and the uniqueness of GRV:

\begin{theorem}[Optimality of GRV] \label{thm:unique_fair_cap}
The smallest solution to (\ref{eq:grve}) exists, and hence the generalized reservation value is well defined (Def.~\ref{def:grv}). Given the current state $(x, s^{i-1}, i)$, and $\sigma := \sigma_i(s^{i-1}, i)$: 

\squishlist
    \item If the generalized reservation value $\sigma > x$, all optimal stopping strategies proceeds. 
    \item If $\sigma < x$, all optimal stopping strategies terminates.
    \item If $\sigma = x$, there exists an optimal stopping time $\tau^*(x, s^{i-1}, i) \geq i$.
\squishend

\end{theorem}

This theorem implies the correctness of our algorithm, as it aligns with the optimal stopping time characterized above. Here, the optimal stopping time is indifferent between proceed or stop when the current max reward equals the GRV of the next box in line. To ensure the uniqueness of the optimal stopping time, WLOG, we adopt the convention that the optimal strategy $ \tau^*(\sigma_i, s^{i-1}, i) $ always stops at $ b_{i-1} $.  

Next, we introduce a lemma for efficiently computing the equivalent reward table for any given state, a key component in our algorithm for calculating the GRV. 

\begin{lemma}\label{lem:computing_rv}
    There is an efficient algorithm (Alg.~\ref{alg:app:gw_1path}) that computes $\phi (x,s, i)$ for all $i$, $x$ and $s$. 
\end{lemma}

For more details, we defer the readers to Appendix~\ref{sec:app:details_exact_opt}. The proof's intuition is to first evaluate $\phi$ at the last box of the hyperbox and then backtrack to the first.

\noindent Now we are ready to present the main theorem of this section, that our algorithm can be implemented in polynomial time, and maximizes the expected payoff. 

\begin{theorem}[Generalized Reservation Rule is Optimal for Single Hyperbox]\label{thm:GMR_path}
Algorithm \ref{alg:pd_1path} is optimal for the Markovian Pandora’s Box problem under single-line precedence graph
and can be computed in polynomial time and space.
\end{theorem}

\begin{proof}
The correctness follows from Thm~\ref{thm:unique_fair_cap}. 

For the running time analysis, we first show that the reservation lookup can be done in polytime, notice that $\hat{\sigma} := \sigma_i(s)$ satisfies that for any $x > \hat{\sigma}$, $\phi (x, s, i) = x$; and for any $x < \hat{\sigma}$, $\phi (x, s, i) \geq x$, we could binary search on the range of $x$ to recover the generalized reservation value. Since the $\phi$ table only requires polynomial time (Lem.~\ref{lem:computing_rv}), the overall running time is in polynomial. 

Regarding the space, storing the $\phi$ tables requires polynomial space, and as the algorithm proceed, the algorithm only need to store the max reward. Thus Algorithm \ref{alg:pd_1path} requires polynomial space. 
\end{proof}


\section{Exact Optimization for Multiple Lines}\label{sec:multi_boxes}

In this section, we will show that the optimal strategy is probing the hyperboxes according to the current GRV of their available box (i.e., the first unopened box). Notice that this strategy becomes fully adaptive for the setting where the precedence graph consists of multiple lines. 

We begin by presenting the definition of Pandora's box with random cost, with reward and cost correlated.

\begin{definition}[Pandora's box with Random Cost]\label{def:pandora_box_random_cost}

A box $b$ is classified as a \textit{Pandora’s Box with Random Cost} if it has a hidden reward $R$ and an opening cost $c$, where both $R$ and $c$ are \textit{random variables} drawn from \textit{known} distributions $\mathcal{D}_R$ and $\mathcal{D}_c$, respectively. Here the opening cost $c$ varies and is \textit{correlated} with the reward $R$ under a \textit{joint distribution} $\Gamma$. 
\end{definition}
Next, we introduce how to equivalently represents a hyperbox as a single box, where the GRV of the hyperbox equals the GRV of the equivalent box. More details in Appendix.~\ref{sec:app:multi_line}.

\begin{lemma}[Equivalent Single Box for Hyperbox]\label{lem:equ_box}
For a stopping time $\tau$ and a hyperbox $\hyperb := \{b_1, \ldots, b_n\}$, there exists a box $\hat{b}$ with random cost (Def.~\ref{def:pandora_box_random_cost}) such that following $\tau$ over $\hyperb$ has the same utility distribution as $\hat{b}$.
\end{lemma}



    

We begin by presenting a key lemma for our main theorem, which establishes that under certain Markovian correlations, the GRV remains an optimal decision rule.

\begin{lemma}[Probing Equivalent Boxes]\label{lem:three_box_lemma}
Given three Pandora's boxes $A, B, C$ with random cost (Def.~\ref{def:pandora_box_random_cost}), with the following property:
\squishlist
    \item The reward and cost of $B$ is independent of the reward and cost of $A$.
    \item The reward and cost (hence payoff) of $C$ depends on both $A$ and $B$ in a Markovian fashion.
    \item The reservation value $\sigma (A) > \sigma (B)> \sigma(C)$, given any realizations of $A$ and $B$ \footnote{Here $\sigma(C)$ is a random variable depends on $A$ and $B$. }, i.e., for any possible value of $x$ of $A$ and possible value of $y$ of $B$: 
    \begin{align*}
        \sigma (A) \geq \sigma (B) \geq [\sigma(C) | R_{A} = x, R_{B} = y]
    \end{align*}
    \item We have a precedence constraint that $A$ and $B$ must be probed before $C$. 
\squishend
then conditioned on any competitive reward $X$, the optimal probing strategy is $A \prec B \prec C$. 
\end{lemma}
We show this lemma by first principles, where we compare the utility of the ordering $A \prec B \prec C$ with that of $B \prec A \prec C$ through a case-by-case analysis of 9 outcomes from the joint distribution of $A$, $B$, and $C$, and aggregate them by the law of total expectation. For more details on the complete proof, please refer to Appendix.~\ref{sec:app:multi_line}.

We are now ready to establish the correctness of GRV for multi-line setting, showing that probing boxes based on their latest GRV maximizes the expected payoff. 

\begin{theorem}[Generalized Reservation Value for Multi-Line Markovian Pandora's Box]\label{thm:ML_MPB}
Given a Markovian Pandora’s Box instance whose precedence graph consists of $m$ lines, the optimal strategy for maximizing expected utility is to probe the hyperboxes based on its latest GRV.
\end{theorem}

\begin{proof}
We denote $n^*$ as the hyperbox with the highest GRV when no boxes are opened, and we denote $\sigma^0_{n^*}$ as the value of this GRV. We denote $n \neq n^*$ as another fixed hyperbox. WLOG, we denote the GRV of the hyperbox as the GRV of the first unopened box conditioned on the realized reward. Let $\pi^*$ denote the fully adaptive strategy that probes the hyperbox with the highest current GRV. We want to show that there doesn't exist any strategy that could outperform $\pi^{*}$ by induction on the number of boxes. 

\textbf{Base Case}: When there are only two boxes, it's immediate to see that $\pi^*$ is the optimal strategy.

\textbf{Induction Step}: Suppose $\pi^*$ is optimal for $q-1$ boxes; we show it remains optimal for $q$ boxes. If the optimal strategy starts by probing the hyperbox $n^*$, then by induction, $\pi^*$ is already optimal for $q$ boxes.

Now, suppose the best strategy $\hat{\pi}$ starts with a hyperbox $n$ whose GRV $\sigma_n^{0}$ is lower than $\sigma_{n^*}^{0}$. By the induction hypothesis, $\hat{\pi}$ follows $\pi^*$ after probing its first hyperbox. The order of $\hat{\pi}$ is as follows: probe $n$ until the next box’s GRV falls below $\sigma_{n^*}^{0}$, then switch to $n^*$, and finally proceed optimally through the remaining boxes.

We construct a new strategy $\bar{\pi}$ with strictly higher expected utility, contradicting the optimality of $\hat{\pi}$. The order of $\bar{\pi}$ is: probe $n^*$ first, then switch to $n$ when the next box’s GRV falls below $\sigma_{n^*}^{0}$, and finally continue optimally as in $\hat{\pi}$. 

Notice that the GRV for the segments containing $n^*$ is $\sigma_{n^*}^{0}$, but the GRV for the segments containing $n$ is smaller than $\sigma_n^{0}$, and that both strategies explore the same segments of $n^*$ and $n$ but in different orders, we are able to apply Lemma~\ref{lem:three_box_lemma} to show $\bar{\pi}$ yields a higher expected payoff than $\hat{\pi}$. This lead to a contradiction. Thus, by induction, we show $\pi^*$ remains optimal for any number of boxes.
\end{proof}

Finally, we present GRV can be implemented in polynomial time and space. Given the optimality of the GRV under the optimal policy, each box's GRV depends only on the realized rewards within its hyperbox, enabling the payoff table to compute GRVs of all states \emph{without} iterative updates to the adaptive order. More details in Appendix.~\ref{sec:app:multi_line}.

\begin{theorem}[Generalized Reservation Value, Multi Lines]\label{thm:GMR_path}
GRV for multi-line setting can be implemented in polynomial time and space. 
\end{theorem}

\section{Exact Optimization for Forest}\label{sec:forest}
In this section, we present how to use previous solutions from the multi-line setting to derive a optimal fully adaptive strategy (i.e., probing according to GRV) for tree and forest. We first present the definitions for minimal tree. 

\begin{definition}[Tree, Branch Vertex, Minimal Tree]\label{def:tree}
A directed \emph{tree} $\T$ is a connected DAG where removing any edge disconnects the graph. A directed \emph{forest} is a DAG whose components are trees. A \emph{branch} vertex in a directed tree is a vertex with at least $2$ outgoing edges. A \emph{minimal tree} is a tree whose strict subgraphs contain no branch vertices.
\end{definition}

The key intuition behind our algorithm is that in a minimal tree, probing the root $r$ reduces the induced graph $\bar{G}$ to \textit{multiple lines}, where the optimal solution is \emph{known}. By extending Lem.~\ref{lem:equ_box}, we contract $\bar{G}$ into a single box with random cost $\hat{b}$ and use the following to calculate the equivalent reward for $r$ given the current reward $x$:
  \begin{align*}
     \Phi(x,r) = & \max\{x, \E[\max\{R_r, x\} - c_r], \\
     & \quad \E[\max\{R_r, x, R_{\hat{b}}\} - c_r - c_{\hat{b}}]\}
 \end{align*}
The terms in the max operator correspond to the utility of not opening $r$, opening only $r$, and opening $r$ while optimally exploring the remaining nodes, respectively. The GRV is fully specified once $\phi$ is known.

Next, we present our algorithm for computing the GRV for every box, which iteratively finds and contracts the minimal trees of the underlying graph into a line, and thus eliminating all the branch vertices iteratively. \footnote{In the presence of opened boxes, for any unopened box $i$, we apply the algorithm to the subtree rooted at $i$ to compute its GRV.}

\begin{algorithm}[!ht]\caption{Updating GRV, Forest}\label{alg:pd_forest}
\begin{algorithmic}[1]
    \Require An instance of Markovian Pandora's box with precedence graph $G$ as a forest. 
    \State $\hat{G} \gets G$, $t \gets 1$.
    \While {there are minimal trees in graph $G_0$}
    \For {every minimal tree $\T_i$ with root $r_i$.}
        \For {every possible state $R_{r_{i}}$ of $r_i$. }
        \State Condition on $R_{r_{i}}$, compute $\phi$ and GRV $\sigma$ for every possible state for vertices/boxes in $\T_i \setminus \{r_i\}$.
        \State Contract $\T_i \setminus \{r_i\}$ into one vertex $\hat{v}_i$, compute its reward and cost distribution condition on $R_{r_{i}}$. 
        \EndFor
        \State Update $\hat{G}$ accordingly. 
    \EndFor 
    \EndWhile
    \State Compute the GRV of remaining boxes.
\end{algorithmic} 
\end{algorithm}

More details regarding the optimality and polynomial runtime of probing according to GRV is in App.~\ref{sec:app:omit_forest}. 


\newpage

\section{Solutions for Static Transition}\label{sec:rapid_mixing}

In this section, we propose a more computationally efficient solution under the static transition assumption, where directed edges within the same component share a common probability transition matrix. Rather than backtracking through \emph{all} descendants to compute the GRV \emph{bottom up}, as in previous solutions, our approach propagates information \emph{top down} from the current box. This significantly improves efficiency as the number of boxes grows.

\subsection{Pandora's Box under Static Transition}\label{subsec:ass_box_static transition}
In this section, we introduce the formal model of Pandora's box under static transition. 

We begin by presenting the definition of static transition: 

\begin{assumption}[Pandora's Box with Static Transition]\label{ass:MC_cost}
Under the setting and notations of Prob.~\ref{prob:mar_pandora}. We further assume the boxes from the same component \footnote{A component of a directed graph is a set of vertices where each vertex can reach every other in the set via a directed path.} has the \emph{same} probability of transitioning from $v_j$ from $v_i$ is determined by a box-independent constant $p_{i,j}$.
\end{assumption}

Next, we formally define the transition matrix.
\begin{definition}[Transition Matrix]\label{def:state_transition}
Given that every box could take one of the $k$ possible values $\{v_1,\ldots, v_k\}$, the transition matrix $\trans:= (p_{i,j})_{k \times k} \in  [0,1]^{k \times k}$. 

Furthermore, for every $j \in [k]$, $\sum_{l\in [k] } p_{jl} = 1$. 
\end{definition}

Next, we outline key assumptions on the transition matrix. For readers familiar with Markov chains, this is equivalent to treating each directed line subgraph of the precedence graph as a Markov chain with rewards as variables and assuming it is irreducible and aperiodic. For details, see Appendix~\ref{subsec:app:MC_preliminary}.

\begin{assumption}[Properties of Transition Matrix]\label{ass:MC_property}
A transition matrix $P$ with state space $\{v_1, \ldots, v_k\}$ satisfies: 
\squishlist
    \item \textbf{Irreducibility}, if for all $i,j \in [k]$, there exists an integer $n \geq 1$ such that: $P^n_{i,j} > 0,$ where $P^n_{i,j}$ denotes the $(i,j)$-entry of $P^n$, representing the probability of transitioning from state $v_i$ to state $v_j$ in $n$ steps.

    \item \textbf{Aperiodic}, if for every $i \in [k]$, the greatest common divisor (gcd) of the set $\{n \geq 1 : P^n_{i,j} > 0\}$ is 1, i.e., $ d = \text{gcd}\{n \geq 1 : P^n_{i,j} > 0\}$, and matrix P is aperiodic if $d = 1$ for all $v_i \in \{v_1, \ldots, v_k\}$.
\squishend
\end{assumption}

Under these two assumptions, we are able to show that for every line as a subgraph of the precedence graph, the (unconditional) reward distribution along probing along this line converges to a unique stable distribution (Lem.~\ref{lem:app:convege_line_stable}). Furthermore, this stable distribution is independent of the reward distribution of the first box in line.  

\newpage
\subsection{Static Transition Under Multiple Lines}\label{subsec:1_hyperbox_rapid_mix}

In this section, we study the case where the precedence graph consists of multiple lines. In the case of static cost, the equivalent reward table can be computed more efficiently using the following theorem.

\begin{theorem}[Markov Chain with Static Cost]\label{thm:MC_constant_cost}
Under Ass.~\ref{ass:MC_cost}, where the reward of each line follows an (infinite) Markov Chain (Def.~\ref{def:app:MC}) and each probe incurs a constant cost $c$. Given that for any $i \in [k-1]$, $p_{i,k}>0$, the equivalent reward $ \phi $ only depends on  $ y $ as the current maximum reward and $ x $ as the current reward. The optimal strategy is to continue if $ \phi(y, x) > y $, and stop if $ \phi(y, x) \leq y $.

\end{theorem}
In this case, the equivalent reward $ \phi $ is \textit{independent} of the current box index. Thus, the equation derived from \textit{Bellman’s principle of optimality} corresponds to a \textit{fixed-point iteration} where the variables are $ \phi $ under different states. We show that this iteration leads to a \textit{unique solution} and converges rapidly. For details, see Appendix~\ref{subsec:app:rapid_mix_MC_static_cost}.

\begin{lemma}[Rapid Convergence of Max Reward]\label{lem:dist_max_reward}
Given any directed line subgraph $\hyperb := \{b_1,\ldots, b_n\}$ of the precedence graph with transition matrix $P_M$, assume that $P_M$ is irreducible and aperiodic, there exist a stable distribution $\pi$, such that given $\delta \in (0,1)$, with probability $1- \delta$, \\ for $t \geq \max \{ 2 \frac{C_{\trans}}{ (\sum_{i = j}^{k} \pi_i) (1 - \alpha_{\trans})},  \frac{\log (1 / \delta)}{\log (1 - (\sum_{i = j}^{k} \pi_i) / 2)} \}$, 
\begin{align*}
    \Pr [R_{\max}(t) \geq v_j] \geq 1- \delta,
\end{align*} 
, for any $j \in [k]$, where we define the max reward of the first $t$ boxes in $\hyperb$ as  $R_{\max} (t) := \max \{R_1,\ldots, R_t\}$.
\end{lemma}

\begin{proof}

From Lem.~\ref{lem:app:convege_line_stable}, given any subset $S \subseteq V = \{v_1,\ldots, v_k\}$, we have for $\pi_S := \sum_{v \in S} \pi_v$:
\begin{align*}
  \pi_S - C_{\trans} \alpha^t_{\trans} \leq \Pr [R_t \in S] \leq \pi_S + C_{\trans} \alpha^t_{\trans} 
\end{align*}

Next, we lower bound the probability that $R_{\max} (t) = v_k$, where we define $\pi_{\min} = \min_{j \in [k]} \pi_j$, $\eps_t := \min \{ \pi_{\min},  C_{\trans} \alpha^t_{\trans} \}$, and $\bar{\eps}_t := \frac{1}{t} \sum_{t_0 \in [t]} \eps_{t_0}$.
\begin{align*}
    &\Pr (R_{\max} (t) = v_k) = 1 - \Pi_{t_0 \in [t]} \Pr[R_t < v_k] \geq \\
 & 1 - \Pi_{t_0 \in [t]} (1 - \pi_k + \eps_{t_0}) \geq 1 - (1 - \pi_k)^t (1 +\frac{ \bar{\eps}_t} { 1 - \pi_k})^t
\end{align*}
where the first equality follows by definition, the second inequality follows by the convergence to stable distribution, the third inequality follows from AM-GM (Lem.~\ref{lem:app:AM_GM}).

Next we upper bound $\bar{\eps}_t$:
\begin{align*}
    & \bar{\eps}_t =  \frac{1}{t} \sum_{t_0 \in [t]} \eps_{t_0} 
    \leq  \frac{1}{t} \sum_{t_0 =1}^{\infty} C_{\trans} \alpha_{\trans}^{t_0} \
    =  \frac{1}{t} C_{\trans} \frac{1}{1 - \alpha_{\trans}}
\end{align*}
Thus, $\Pr[R_{\max} (t) = v_k] \geq 1 - [1 - \pi_k + \frac{1}{t} \frac{C_{\trans}}{(1 - \alpha_{\trans})}]^t$.

Notice that for $t \geq 2 \frac{C_{\trans}}{ \pi_k (1 - \alpha_{\trans})}$, we have:
\begin{align*}
    (1 - \pi_k + \frac{1}{t} \frac{C_{\trans}}{1 - \alpha_{\trans}}) \leq 1 - \pi_k / 2
\end{align*}
Given $\delta>0$, for $t \geq \max \{ 2 \frac{C_{\trans}}{ \pi_k (1 - \alpha_{\trans})}, \frac{ \log (1 / \delta)}{\log (1 - \pi_k / 2)} \}$,  
\begin{align*}
    & \Pr[R_{\max} (t) = v_k] \geq 1 - (1 - \pi_k + \frac{1}{t} \frac{C_{\trans}}{1 - \alpha_{\trans}})^t \\
    & \geq 1 - (1 - \pi_k/2)^t \geq 1 - \delta
\end{align*}

Following the same logic, we could lower bound the probability that $R_{\max} \geq v_j$.
\end{proof}

The lemma implies that when the precedence graph is a directed line, there exists a near optimal policy that explores at most $\tilde{\Theta}(1)$ boxes. This result directly leads to the following theorem. For more details, please refer to Appendix~\ref{subsec:app:multi_line_general}.

\begin{theorem}[Near Optimal Probing, Single Line]\label{thm:greedy_optimal_single_line}
Given a Markovian Pandora’s box with static transition, where the precedence graph is a single line $ \hyperb $ and the transition matrix $ \trans $ is irreducible, aperiodic, and associated with a stable distribution $ \pi $. There exists constant $C_{\trans}>0$ and $\alpha_{\trans} \in (0,1)$, that, given any $\delta \in (0,1)$, let $t_{\delta} = \max \{ 2 \frac{C_{\trans}}{ \pi_k (1 - \alpha_{\trans})},  \frac{\log (1 / \delta)}{\log (1 - \pi_k / 2)} \}$. Then, the expected utility from optimal probing on the subgraph $ \hyperb|t_{\delta} $, containing only the first $ t_{\delta} $ boxes of $ \hyperb $, is near-optimal, i.e.,
\begin{align*}
   \OPT (\hyperb|{t_\delta}) \geq   \OPT(\hyperb) - 2\delta v_k
\end{align*}
where $\OPT(\hyperb)$ is the expected utility from optimal probing graph $\hyperb$.
\end{theorem}

This theorem implies that it's near optimal to truncate any hyperbox by its first $t_\delta$ steps, which we generalize to the multi-line setting. More details in App.~\ref{subsec:app:multi_line_general}.

\begin{theorem}[Near Optimal Probing, Multi Lines]\label{thm:greedy_optimal_multi_line}
Given a Markovian Pandora’s box with static transition, where the precedence graph $G$ consists of $q$ directed lines $ \hyperb_1, \ldots, \hyperb_q$. Suppose every transition matrix $P_j$ of each $\hyperb_j, j \in [q]$ are irreducible, aperiodic, and associated with a stable distribution $ \pi(j) $. There exists constant $C>0$ and $\alpha \in (0,1)$, that, given any $\delta \in (0,1)$, let $t_{\delta} = \max \{ 2 \frac{C}{ \pi^* (1 - \alpha)},  \frac{\log (1 / \delta)}{\log (1 - \pi^* / 2)} \}$, where $\pi^* = \max_{j \in [q]} \pi (j)_k$. Then, the expected utility from optimal probing on the subgraph $ \cup_{j \in [q]} (\hyperb_j |t_{\delta}) $, containing only the first $ t_{\delta} $ boxes of each hyperbox, is near-optimal, i.e.,
\begin{align*}
   \OPT [\cup_{j \in [q]} (\hyperb_j |t_{\delta})] \geq   \OPT(G) - 2q\delta v_k
\end{align*}
where $\OPT(G)$ is the expected utility from optimal probing graph $G$. 
\end{theorem}

\subsection{Extension To Forest}\label{subsec:transition_forest}

In this section, we show that for the forest setting, there still exist an induced subgraph with a much smaller size that guarantees roughly $1/2$-approximation of the expected payoff. We first define the approximate ratio below:

\begin{definition}[Approximate Ratio]\label{def:approx_ratio}
    A strategy $\pi$ is said to be a $C$-approximation of the Markovian Pandora's box, if, given any strategy $\hat{\pi}$, the expected utility of $\pi$ satisfies:
    \begin{align*}
        \E [\max_{i \in \mathcal{O}(\pi)} R_i - \sum_{i \in \mathcal{O}(\pi)} c_i] \geq C \cdot  \E[\max_{i \in \mathcal{O}(\hat{\pi})} R_i] - \sum_{i \in \mathcal{O}(\hat{\pi})} c_i].
    \end{align*}
    where $\mathcal{O}(\pi)$ denotes the random set selected by strategy $\pi$. 
\end{definition}

Next, we present a lemma that there exists a non-adaptive strategy that guarantees the utility of best adaptive strategy.\footnote{As proven in Thm.3~\cite{bfll20}, this originally holds for order constraints with independent boxes, and their proof generalizes easily to our setting with a forest precedence graph.}

\begin{lemma}[Adaptivity Gap~\cite{bfll20}]\label{lem:forest_adaptivity}
    Consider the Markovian Pandora's box with a forest precedence graph, then for every adaptive strategy $\hat{\pi}$, there exists a non-adaptive strategy $\pi$ that obtains $1/2$ approximation. 
\end{lemma}

Next we are ready to present our results on a (rougly) $1/2$ approximation for the forest setting. 

\begin{theorem}[$1/2$ Approximation, Forest]\label{thm:forest_approx}

Given a Markovian Pandora’s box with static transition, where the precedence graph $G$ is a forest. Given any $\delta \in (0,1)$,there exists an algorithm such that it finds a best fixed line subgraph $\hyperb$ within $\Delta(G)^{\tilde{\Theta}(1)}$ time, such that given any alternative adaptive policy $\pi$:
    \begin{align*}
        \E [\max_{i \in \mathcal{O}(\pi)} R_i - &\sum_{i \in \mathcal{O}(\pi)} c_i] \geq \\
        &1/2 \cdot  \E[\max_{i \in \mathcal{O}(\hat{\pi})} R_i] - \sum_{i \in \mathcal{O}(\hat{\pi})} c_i] -q \delta.
    \end{align*}
where $\Delta(G)$ is the degree and $q$ is the number of trees in $G$. 
    
\end{theorem}

\section{Conclusions}\label{sec:conclusion}

In this work, we introduced the \emph{Markovian Pandora’s Box} problem, extending the classical framework to incorporate \emph{structural constraints} and \emph{probabilistic dependencies}. We developed the first \emph{optimal algorithm} for this problem on a \emph{forest-structured graph}, demonstrating that despite its fully adaptive nature, the solution can be computed efficiently in \emph{polynomial time and space}.  Furthermore, under \emph{static transition}, we derived \emph{faster} algorithms via \emph{subgraph optimization}, ensuring near-optimal performance while significantly reducing computational complexity. Our results provide new insights into constrained sequential exploration with Markovian correlations in search and selection problems.






\ifdefined\isarxiv
\newcommand{\etalchar}[1]{$^{#1}$}

\else
\bibliography{ref}

\begin{thebibliography}{DGMM22}

\bibitem[ACC{\etalchar{+}}11]{acc+11}
Nir Ailon, Bernard Chazelle, Kenneth~L Clarkson, Ding Liu, Wolfgang Mulzer, and C~Seshadhri.
\newblock Self-improving algorithms.
\newblock {\em SIAM Journal on Computing}, 40(2):350--375, 2011.

\bibitem[AFRT24]{afrt24}
Georgios Amanatidis, Federico Fusco, Rebecca Reiffenh{\"a}user, and Artem Tsikiridis.
\newblock Pandora's box problem over time.
\newblock {\em arXiv preprint arXiv:2407.15261}, 2024.

\bibitem[AJS20]{ajs20}
Ali Aouad, Jingwei Ji, and Yaron Shaposhnik.
\newblock The pandora's box problem with sequential inspections.
\newblock {\em Available at SSRN 3726167}, 2020.

\bibitem[AKL{\etalchar{+}}19]{akl+19}
Daniel Alabi, Adam~Tauman Kalai, Katrina Liggett, Cameron Musco, Christos Tzamos, and Ellen Vitercik.
\newblock Learning to prune: Speeding up repeated computations.
\newblock In {\em Conference on Learning Theory}, pages 30--33. PMLR, 2019.

\bibitem[Arm17]{a17}
Mark Armstrong.
\newblock Ordered consumer search.
\newblock {\em Journal of the European Economic Association}, 15(5):989--1024, 2017.

\bibitem[ASW16]{asw16}
Marek Adamczyk, Maxim Sviridenko, and Justin Ward.
\newblock Submodular stochastic probing on matroids.
\newblock {\em Mathematics of Operations Research}, 41(3):1022--1038, 2016.

\bibitem[BB12]{bb12}
James Bergstra and Yoshua Bengio.
\newblock Random search for hyper-parameter optimization.
\newblock {\em Journal of machine learning research}, 13(2), 2012.

\bibitem[BBS18]{bbs18}
Dirk Bergemann, Alessandro Bonatti, and Alex Smolin.
\newblock The design and price of information.
\newblock {\em American economic review}, 108(1):1--48, 2018.

\bibitem[BC23]{bc23}
Hedyeh Beyhaghi and Linda Cai.
\newblock Pandora’s problem with nonobligatory inspection: Optimal structure and a ptas.
\newblock In {\em Proceedings of the 55th Annual ACM Symposium on Theory of Computing}, pages 803--816, 2023.

\bibitem[BC24]{bc24}
Hedyeh Beyhaghi and Linda Cai.
\newblock Recent developments in pandora's box problem: Variants and applications.
\newblock {\em ACM SIGecom Exchanges}, 21(1):20--34, 2024.

\bibitem[BDD{\etalchar{+}}19]{bdd+19}
Maria-Florina Balcan, Dan DeBlasio, Travis Dick, Carl Kingsford, Tuomas Sandholm, and Ellen Vitercik.
\newblock How much data is sufficient to learn high-performing algorithms?
\newblock {\em Proceedings of the 53rd Annual ACM SIGACT Symposium on Theory of Computing (STOC) 2021}, 2019.

\bibitem[BDSV18]{bdsv18}
Maria-Florina Balcan, Travis Dick, Tuomas Sandholm, and Ellen Vitercik.
\newblock Learning to branch.
\newblock In {\em International conference on machine learning}, pages 344--353. PMLR, 2018.

\bibitem[BEFF24]{beff24}
Ben Berger, Tomer Ezra, Michal Feldman, and Federico Fusco.
\newblock Pandora’s problem with deadlines.
\newblock In {\em Proceedings of the AAAI Conference on Artificial Intelligence}, volume~38, pages 20337--20343, 2024.

\bibitem[BFLL20]{bfll20}
Shant Boodaghians, Federico Fusco, Philip Lazos, and Stefano Leonardi.
\newblock Pandora's box problem with order constraints.
\newblock In {\em Proceedings of the 21st ACM Conference on Economics and Computation}, pages 439--458, 2020.

\bibitem[BK19]{bk19}
Hedyeh Beyhaghi and Robert Kleinberg.
\newblock Pandora's problem with nonobligatory inspection.
\newblock In {\em Proceedings of the 2019 ACM Conference on Economics and Computation}, pages 131--132, 2019.

\bibitem[BNVW17]{bnvw17}
Maria-Florina Balcan, Vaishnavh Nagarajan, Ellen Vitercik, and Colin White.
\newblock Learning-theoretic foundations of algorithm configuration for combinatorial partitioning problems.
\newblock In {\em Conference on Learning Theory}, pages 213--274. PMLR, 2017.

\bibitem[CCHS24]{cchs24}
Shuchi Chawla, Dimitris Christou, Amit Harlev, and Ziv Scully.
\newblock Combinatorial selection with costly information.
\newblock {\em arXiv preprint arXiv:2412.03860}, 2024.

\bibitem[CDKT19]{cdkt19}
Shuchi Chawla, Shaleen Deep, Paraschos Koutrisw, and Yifeng Teng.
\newblock Revenue maximization for query pricing.
\newblock {\em Proceedings of the VLDB Endowment}, 13(1):1--14, 2019.

\bibitem[CFG{\etalchar{+}}00]{cfg+00}
Moses Charikar, Ronald Fagin, Venkatesan Guruswami, Jon Kleinberg, Prabhakar Raghavan, and Amit Sahai.
\newblock Query strategies for priced information.
\newblock In {\em Proceedings of the thirty-second annual ACM symposium on Theory of computing}, pages 582--591, 2000.

\bibitem[CGMT21]{cgmt21}
Shuchi Chawla, Evangelia Gergatsouli, Jeremy McMahan, and Christos Tzamos.
\newblock Approximating pandora's box with correlations.
\newblock {\em arXiv preprint arXiv:2108.12976}, 2021.

\bibitem[CGT{\etalchar{+}}20]{cgt+20}
Shuchi Chawla, Evangelia Gergatsouli, Yifeng Teng, Christos Tzamos, and Ruimin Zhang.
\newblock Pandora's box with correlations: Learning and approximation.
\newblock In {\em 2020 IEEE 61st Annual Symposium on Foundations of Computer Science (FOCS)}, pages 1214--1225. IEEE, 2020.

\bibitem[CHKK15]{chkk15}
Yuxin Chen, S~Hamed Hassani, Amin Karbasi, and Andreas Krause.
\newblock Sequential information maximization: When is greedy near-optimal?
\newblock In {\em Conference on Learning Theory}, pages 338--363. PMLR, 2015.

\bibitem[CJK{\etalchar{+}}15]{cjk+15}
Yuxin Chen, Shervin Javdani, Amin Karbasi, J~Bagnell, Siddhartha Srinivasa, and Andreas Krause.
\newblock Submodular surrogates for value of information.
\newblock In {\em Proceedings of the AAAI Conference on Artificial Intelligence}, volume~29, 2015.

\bibitem[CMS10]{cms10}
Kenneth~L Clarkson, Wolfgang Mulzer, and C~Seshadhri.
\newblock Self-improving algorithms for convex hulls.
\newblock In {\em Proceedings of the Twenty-First Annual ACM-SIAM Symposium on Discrete Algorithms}, pages 1546--1565. SIAM, 2010.

\bibitem[DFH{\etalchar{+}}23]{dfh+23}
Bolin Ding, Yiding Feng, Chien-Ju Ho, Wei Tang, and Haifeng Xu.
\newblock Competitive information design for pandora's box.
\newblock In {\em Proceedings of the 2023 Annual ACM-SIAM Symposium on Discrete Algorithms (SODA)}, pages 353--381. SIAM, 2023.

\bibitem[DGMM22]{dgmm22}
Mahsa Derakhshan, Negin Golrezaei, Vahideh Manshadi, and Vahab Mirrokni.
\newblock Product ranking on online platforms.
\newblock {\em Management Science}, 68(6):4024--4041, 2022.

\bibitem[Dov18]{d18}
Laura Doval.
\newblock Whether or not to open pandora's box.
\newblock {\em Journal of Economic Theory}, 175:127--158, 2018.

\bibitem[FLL23]{fll23}
Hu~Fu, Jiawei Li, and Daogao Liu.
\newblock Pandora box problem with nonobligatory inspection: Hardness and approximation scheme.
\newblock In {\em Proceedings of the 55th Annual ACM Symposium on Theory of Computing}, pages 789--802, 2023.

\bibitem[GGM06]{ggm06}
Ashish Goel, Sudipto Guha, and Kamesh Munagala.
\newblock Asking the right questions: Model-driven optimization using probes.
\newblock In {\em Proceedings of the twenty-fifth ACM SIGMOD-SIGACT-SIGART symposium on Principles of database systems}, pages 203--212, 2006.

\bibitem[GJSS19]{gjss19}
Anupam Gupta, Haotian Jiang, Ziv Scully, and Sahil Singla.
\newblock The markovian price of information.
\newblock In {\em Integer Programming and Combinatorial Optimization: 20th International Conference, IPCO 2019, Ann Arbor, MI, USA, May 22-24, 2019, Proceedings 20}, pages 233--246. Springer, 2019.

\bibitem[GK01]{gk01}
Anupam Gupta and Amit Kumar.
\newblock Sorting and selection with structured costs.
\newblock In {\em Proceedings 42nd IEEE Symposium on Foundations of Computer Science}, pages 416--425. IEEE, 2001.

\bibitem[GKSW24]{gksw24}
Khashayar Gatmiry, Thomas Kesselheim, Sahil Singla, and Yifan Wang.
\newblock Bandit algorithms for prophet inequality and pandora's box.
\newblock In {\em Proceedings of the 2024 Annual ACM-SIAM Symposium on Discrete Algorithms (SODA)}, pages 462--500. SIAM, 2024.

\bibitem[GN13]{gn13}
Anupam Gupta and Viswanath Nagarajan.
\newblock A stochastic probing problem with applications.
\newblock In {\em Integer Programming and Combinatorial Optimization: 16th International Conference, IPCO 2013, Valparaiso, Chile, March 18-20, 2013. Proceedings 16}, pages 205--216. Springer, 2013.

\bibitem[GNS16]{gns16}
Anupam Gupta, Viswanath Nagarajan, and Sahil Singla.
\newblock Algorithms and adaptivity gaps for stochastic probing.
\newblock In {\em Proceedings of the twenty-seventh annual ACM-SIAM symposium on Discrete algorithms}, pages 1731--1747. SIAM, 2016.

\bibitem[GNS17]{gns17}
Anupam Gupta, Viswanath Nagarajan, and Sahil Singla.
\newblock Adaptivity gaps for stochastic probing: Submodular and xos functions.
\newblock In {\em Proceedings of the Twenty-Eighth Annual ACM-SIAM Symposium on Discrete Algorithms}, pages 1688--1702. SIAM, 2017.

\bibitem[GR16]{gr16}
Rishi Gupta and Tim Roughgarden.
\newblock A pac approach to application-specific algorithm selection.
\newblock In {\em Proceedings of the 2016 ACM Conference on Innovations in Theoretical Computer Science}, pages 123--134, 2016.

\bibitem[GT22]{gt22}
Evangelia Gergatsouli and Christos Tzamos.
\newblock Online learning for min sum set cover and pandora’s box.
\newblock In {\em International Conference on Machine Learning}, pages 7382--7403. PMLR, 2022.

\bibitem[GT23]{gt23}
Evangelia Gergatsouli and Christos Tzamos.
\newblock Weitzman's rule for pandora's box with correlations.
\newblock {\em arXiv preprint arXiv:2301.13534}, 2023.

\bibitem[HKY18]{hky18}
Elad Hazan, Adam Klivans, and Yang Yuan.
\newblock Hyperparameter optimization: a spectral approach.
\newblock In {\em International Conference on Learning Representations}, 2018.

\bibitem[JT16]{jt16}
Kevin Jamieson and Ameet Talwalkar.
\newblock Non-stochastic best arm identification and hyperparameter optimization.
\newblock In {\em Artificial intelligence and statistics}, pages 240--248. PMLR, 2016.

\bibitem[KK18]{kk18}
Jon Kleinberg and Robert Kleinberg.
\newblock Delegated search approximates efficient search.
\newblock In {\em Proceedings of the 2018 ACM Conference on Economics and Computation}, pages 287--302, 2018.

\bibitem[KLBL17]{kll17}
Robert Kleinberg, Kevin Leyton-Brown, and Brendan Lucier.
\newblock Efficiency through procrastination: Approximately optimal algorithm configuration with runtime guarantees.
\newblock In {\em IJCAI}, volume~3, page~1, 2017.

\bibitem[KTH{\etalchar{+}}19]{kthh+19}
Lars Kotthoff, Chris Thornton, Holger~H Hoos, Frank Hutter, and Kevin Leyton-Brown.
\newblock Auto-weka: Automatic model selection and hyperparameter optimization in weka.
\newblock {\em Automated machine learning: methods, systems, challenges}, pages 81--95, 2019.

\bibitem[KWW16]{kww16}
Robert Kleinberg, Bo~Waggoner, and E~Glen Weyl.
\newblock Descending price optimally coordinates search.
\newblock In {\em Proceedings of the 2016 ACM Conference on Economics and Computation}, pages 23--24, 2016.

\bibitem[LJD{\etalchar{+}}17]{ldrt17}
Lisha Li, Kevin Jamieson, Giulia DeSalvo, Afshin Rostamizadeh, and Ameet Talwalkar.
\newblock Hyperband: A novel bandit-based approach to hyperparameter optimization.
\newblock {\em The Journal of Machine Learning Research}, 18(1):6765--6816, 2017.

\bibitem[LL22]{ll22}
Jian Li and Daogao Liu.
\newblock Multi-token markov game with switching costs.
\newblock In {\em Proceedings of the 2022 Annual ACM-SIAM Symposium on Discrete Algorithms (SODA)}, pages 1780--1807. SIAM, 2022.

\bibitem[LP17]{lpw17}
David~A Levin and Yuval Peres.
\newblock {\em Markov chains and mixing times}, volume 107.
\newblock American Mathematical Soc., 2017.

\bibitem[LS17]{ls17}
Hao Li and Xianwen Shi.
\newblock Discriminatory information disclosure.
\newblock {\em American Economic Review}, 107(11):3363--3385, 2017.

\bibitem[OW15]{o15}
Wojciech Olszewski and Richard Weber.
\newblock A more general pandora rule?
\newblock {\em Journal of Economic Theory}, 160:429--437, 2015.

\bibitem[Sin18]{s18}
Sahil Singla.
\newblock The price of information in combinatorial optimization.
\newblock In {\em Proceedings of the twenty-ninth annual ACM-SIAM symposium on discrete algorithms}, pages 2523--2532. SIAM, 2018.

\bibitem[SLA12]{sla12}
Jasper Snoek, Hugo Larochelle, and Ryan~P Adams.
\newblock Practical bayesian optimization of machine learning algorithms.
\newblock {\em Advances in neural information processing systems}, 25, 2012.

\bibitem[SMV{\etalchar{+}}20]{smv+20}
Prabhu~Teja Sivaprasad, Florian Mai, Thijs Vogels, Martin Jaggi, and Fran{\c{c}}ois Fleuret.
\newblock Optimizer benchmarking needs to account for hyperparameter tuning.
\newblock In {\em International conference on machine learning}, pages 9036--9045. PMLR, 2020.

\bibitem[Wei79]{w79}
Martin~L Weitzman.
\newblock Optimal search for the best alternative.
\newblock {\em Econometrica: Journal of the Econometric Society}, pages 641--654, 1979.

\bibitem[WGS18]{wgs18}
Gell{\'e}rt Weisz, Andras Gyorgy, and Csaba Szepesv{\'a}ri.
\newblock Leapsandbounds: A method for approximately optimal algorithm configuration.
\newblock In {\em International Conference on Machine Learning}, pages 5257--5265. PMLR, 2018.

\end{thebibliography}
\bibliographystyle{neurips2023}

\fi

\appendix

\section{More Details from Literature Review}\label{sec:app:literature_review}

\paragraph{Pandora's Box and Friends}
Pandora's box, originates from~\cite{w79}, has since attracted a lot of research interests in studying its variants and application scenarios. These variants include:

\squishlist
    \item \textbf{Pandora's box with Order Constraints}. Previous work by \citet{bfll20} focused on order constraints, where some boxes must be opened after others, and rewards are independent accross boxes, making it a \emph{special case} of ours. Their optimal strategy is partially adaptive, whereas ours is fully adaptive. Consequently, our setting is fundamentally more challenging than theirs.
    \item \textbf{Pandora's box with Correlations}. Recent studies have shown growing interest in the Pandora’s Box problem with correlations~\cite{cdkt19, cgt+20, cgmt21, gt23}. These works study the cost minimization version of the Pandora’s Box problem and focus on deriving adaptive strategies that approximate the fully adaptive (FA) or partially adaptive optimal solutions. Their results show that approximating the FA optimal within a constant factor is \emph{NP-hard}. In contrast, our setting assumes structured correlations, allowing for the exact optimization of FA strategies within polynomial time.
    \item \textbf{Online Variants}. \citet{gt22} studies an online learning variant of Pandora's Box and Min Sum Set Cover, proposing a computationally efficient algorithm that is constant-competitive with the optimal search order, extending to a bandit setting and generalized selection under matroid constraints. \citet{gksw24} explores \emph{Prophet Inequality} and \emph{Pandora’s Box} in the \emph{Multi-Armed Bandits} model, developing \emph{near-optimal} regret-minimizing algorithms $\tilde{O}(\text{poly}(n)\sqrt{T})$ that balance \emph{exploration and exploitation} by maintaining confidence intervals on the optimal policy’s indices.

    \item \textbf{Nonobligatory Inspection}. Pandora’s Box with nonobligatory inspection is a variant of Weitzman’s Pandora’s problem, introduced by~\citet{d18}, where the searcher is not required to pay the inspection cost before selecting an alternative. Unlike the original problem, this version cannot be solved optimally by a simple ranking-based policy. 
    
    \citet{bk19} provides the \emph{first non-trivial approximation guarantees} for this problem, introducing \emph{committing policies} that are computationally efficient and proving that the optimal committing policy achieves a \(1 - \frac{1}{e} \approx 0.63\) approximation, improving to \( \frac{4}{5} \) for the \emph{two-box case}. \citet{bc23} provides a \emph{structural characterization} of the \emph{optimal policy} for this problem, establishes its \emph{NP-completeness}, and develops a \emph{polynomial-time approximation scheme (PTAS)} using a novel reduction, while also proving a \emph{tight $0.8$-approximation} for committing policies across general distributions. Concurrent work~\cite{fll23} also establishes the \emph{NP-hardness} of computing an optimal policy for this problem, and develops a \emph{polynomial-time approximation scheme (PTAS)} that achieves an expected payoff of at least \((1 - \epsilon)\) of the optimal for any \(\epsilon > 0\).

    \item \textbf{Others}. Several other notable variants of Pandora’s Box have been studied, including settings where boxes can be partially opened~\citep{ajs20}, problems with generalized objective functions~\citep{o15}, models incorporating deadlines~\citep{beff24}, cases with time-dependent rewards and costs~\citep{afrt24}, and scenarios where box information is strategically revealed~\citep{dfh+23}. 

\squishend

\paragraph{Data Driven Algorithm Design} The Pandora’s Box problem provides a fundamental framework for decision-making under costly information, making it particularly relevant to data-driven algorithm design when cost considerations are taken into account. Data-driven algorithm design~\cite{gr16} includes greedy heuristic selection, self-improving algorithms~\citep{cms10,acc+11}, and parameter tuning in optimization and machine learning~\cite{ggm06, bb12,sla12,jt16,bnvw17,ldrt17, kll17, hky18, wgs18, bdsv18, akl+19, bdd+19,kthh+19,smv+20}.

\paragraph{Connections to Other Problems}

Another relevant line of research is the \emph{stochastic probing} problem, which involves deciding both when and which elements to probe. ~\citet{gn13} proposes and examines this problem where elements in a universe are active with given probabilities, and an algorithm must probe elements to determine their activity while satisfying outer and inner packing constraints (such as matroid and knapsack intersections) to maximize total weight. As an application, they provide the first polynomial-time $\Omega(1/k)$-approximate sequential posted price mechanism for $k$-matroid intersection constraints.~\citet{asw16} generalizes the stochastic probing problem by extending the objective from linear to monotone submodular functions and presents a $\frac{(1-1/e)}{k_{\text{in}}+k_{\text{out}}+1}$-approximation algorithm for settings with $k_{\text{in}}$ inner matroid constraints and $k_{\text{out}}$ outer matroid constraints, along with an improved $\frac{1}{k_{\text{in}}+k_{\text{out}}}$-approximation for linear objectives. The studies by \citet{gns16,gns17} focus on the adaptivity gaps of stochastic probing problems, particularly in settings with prefix-closed outer constraints and submodular or XOS objectives. Here the adaptivity gap refers to the ratio between the optimal expected value of the best adaptive policy (which makes decisions dynamically based on observed outcomes) and the best non-adaptive policy (which commits to a fixed decision sequence in advance).

Other relevant problems also exhibit similar information structures and/or solution concepts, such as search~\cite{a17,kk18}, ranking~\cite{dgmm22}, Markov game~\cite{ll22}, sorting and selection~\cite{gk01}, revenue maximization~\cite{kww16, cdkt19} and costly information~\cite{cfg+00, cjk+15, chkk15, ls17, s18, bbs18, gjss19,cchs24}.
\section{More Details from Problem Formulation}\label{sec:app:formulation}

\begin{lemma}[The sub-optimality of PA strategies]\label{lem:app:sub_opt_PA}
Consider the box $A$, $B$ and $C$, whose value denoted as $v_A$, $v_B$, $v_C$, respectively. The box is ordered such that box $A$ must be opened before box $B$. Then there doesn't exist any partially adaptive strategy that is optimal, for the box distribution specified as follows:

\begin{align*}
    v_A  &=  
    \begin{cases}  
        900 \text{ w.p. } 0.1, & 1 \text{ w.p. } 0.9  
    \end{cases}, \\
    v_B &=  
    \begin{cases}  
        v_A + 20 \text{ w.p. } 0.5, & v_A - 10 \text{ w.p. } 0.5  
    \end{cases}, \\
    v_C &=  
    \begin{cases}  
        50 \text{ w.p. } 0.5, & 10 \text{ w.p. } 0.5  
    \end{cases}.
\end{align*}

and for $c_A = 20$, $c_B = 3$, $c_C = 5$.
\end{lemma}

\begin{proof}[Proof Sketch]
It's easy to see that the best PA strategy will order $A \prec C$, so it's sufficient to consider the expected utility of PA strategies of the following two orders: 1)\ $A \prec B \prec C$, and 2)\ $A \prec C \prec B$. From the calculation, order~1) will give us an expected utility of $90$, and order~2) will give us an expected utility of $92.5$, \footnote{We already stop optimally.} The best FA strategy is to first probe $A$, and 1)\ if $v_A = 1$, probe box $C$, 2)\ if $v_A = 900$, probe $B$. Thus, the best adaptive strategy gives a total utility of $92.7$, outperforming the best PA strategy.
\end{proof}

\textbf{Notations}. 
We use $b_i$ and $i$ interchangeably to denote the same box. Following the notations in Prob.~\ref{prob:mar_pandora}: For every $i \in [n]$, we use $s^i$ and $R_i$ interchangeably to denote the reward of box $b_i$, and we denote the probability density function (pdf) of its reward as $\vp_i$, i.e., $\vp_i[s_q] = \Pr [R_i = s_q]$. We use $P_i \in \R_{+}^{K \times K}$ to denote the (probability) transition matrix from $\D_i$ to $\D_{i+1}$, i.e., $\vp_{i+1} = \vp_i \cdot P_{i+1}$. 

We denote $\ve_k$ as the $k$-th basis vector, where $\ve_k [k] = 1$, and $\ve_k[j] = 0$ for $j \neq k$. Suppose the algorithm just opened $b_j$ and observed its state $s_q$, we can update $\vp_j = \ve_q$ as the updated distribution of $b_j$. Using $\vp_{i+1} = \vp_i \cdot P_{i+1}$, we can update $\vp_\ell$ for all $\ell \ge j+1$. 


We use $\tilde{O}$ as a variant of Big-O notation that disregards polylogarithmic factors. When we refer to an algorithm as running in polynomial time, we mean that its running time is polynomial in both the number of possible reward values and the number of distinct boxes.

\section{More Details from Single Line}\label{sec:app:details_exact_opt}

\subsection{More Details for Generalized Reservation Value}

\begin{lemma}[Properties of $\Phi$ and $H_i$]\label{lem:app:phi_prop}
Given any state $(x, s^{i-1}, i)$, 

\squishlist
    \item $\Phi(\cdot, s^{i-1}, i )$ is 1-Lipschitz and monotone non-decreasing.
    \item Let $H_i(x, s^{i-1}) := \Phi (x, s^{i-1}, i) - x$, then $H_i(\cdot, s^{i-1})$ is nonnegative, 1-Lipschitz and monotone non-increasing.
    \item For $z_i$ as the generalized reservation value (Def.~\ref{def:GRV}) of the $i$-th box of the hyperbox, then $\Phi(x,s^{i-1}, i ) = x$ for any $x\ge z_i$.
\squishend
\end{lemma}

\begin{proof}
Given any $a < b$, 
\begin{align*}
    &\Phi (b, s^{i-1}, i) - \Phi (a, s^{i-1}, i)\\
	&\le \E \Big[\max \{b , \max _{j=i} ^{\tau^*(b, s^{i-1}, i)}R_j\} - \max \{a , \max _{j=i} ^{\tau^*(b, s^{i-1}, i)}R_j\}\Big]\\
	& \le b-a
\end{align*}
where in the first inequality, we used that $\tau^*(b, s^{i-1}, i)$ is a suboptimal strategy for $\Phi^\tau(a, s^{i-1}, i)$. Using the same reasoning, we have:
\begin{align*}
    \Phi (b, s^{i-1}, i) &= \E \Big[\max \{b , \max _{j=i} ^{\tau^*(b, s^{i-1}, i)}R_j\}\Big] - \sum_{j=i}^{\tau^*(b, s^{i-1}, i)} c_j \\
    & \geq \E \Big[\max \{b , \max _{j=i} ^{\tau^*(a, s^{i-1}, i)}R_j\}\Big] - \sum_{j=i}^{\tau^*(a, s^{i-1}, i)} c_j\\
    & \geq \E \Big[\max \{a , \max _{j=i} ^{\tau^*(a, s^{i-1}, i)}R_j\}- \sum_{j=i}^{\tau^*(a, s^{i-1}, i)} c_j \Big] = \Phi (a, s^{i-1}, i)
\end{align*}
where the first inequality follows from with the increase of the current realized reward, the optimal stopping rule $\tau^*$ will only stop earlier. Thus, $\Phi (\cdot, s^{i-1}, i)$ is monotone non-decreasing. Now consider $H_i$, we have 
\[
H_i(b, s^{i-1})- H_i(a, s^{i-1}) = \Phi (b, s^{i-1}, i) - \Phi (a, s^{i-1}, i) - (b-a) \le 0
\]
where we use that $\Phi (\cdot, s^{i-1}, i)$ is 1-Lipschitz. The above inequality implies that $H_i(x, s^{i-1})$ is 1-Lipschitz and monotone non-increasing. Lastly, $\Phi(x,s^{i-1}, i ) - x = 0$ for all $x\ge z_i$ follows from the that fact that $z_i$ is the smallest such that $H_i(z_i, s^{i-1}) = 0$ and $H_i$ is non-negative and monotone non-increasing.

\end{proof}

\begin{lemma}[Properties of GRV]\label{lem:app:properties_GRV}
 Given a Pandora's box with line precedence graph $\mathcal{L}=\left[b_1, \ldots, b_n\right]$, the generalized reservation value of every box $i \in [n]$ satisfies the following property: Given any state $s^{i-1}$ as the state of $(i-1)$-th box,
\begin{itemize}
    \item $\sigma_i(s^{i-1}, i)$ is nondecreasing as additional boxes are appended to $\hyperb$.
    \item Let $\eta$ be the (random) index of the first box that has generalized reservation value smaller than $\sigma_i(s^{i-1}, i)$, then $\sigma_i(s^{i-1}, i)$ depends only on the (sub)hyperbox $\hat{\hyperb} := \{b_i, \ldots, b_{\eta}\}$. If $i = \eta$ with probability $1$, then $\sigma_i(s^{i-1}, i)$ depends only on $b_i$. 
\end{itemize}
\end{lemma}

\begin{proof}
\squishlist
    \item The first property holds because the optimal policy stops at the additional boxes only if they yield a higher expected payoff. Consequently, appending boxes at the end of $\hyperb$ can only increase the expected payoff for any given state. As a result, this operation leads to a nondecreasing generalized reservation value.
    \item The second property is due to that the optimal stopping time will stop at $(\eta-1)$-th box, hence the GRV doesn't depend on any boxes starting from $\eta$.
\squishend
\end{proof}

\begin{lemma}\label{lem:app:unique_fair_cap}
The smallest solution to (\ref{eq:grve}) exists, and hence Definition \ref{def:GRV} is well defined.
Given current state $(x, s^{i-1}, i)$, if the generalized reservation value $z_i = x$, then there exists some optimal stopping time $\tau^*(x, s^{i-1}, i) \geq i$. 
\end{lemma}

\begin{proof}
Given any state $(x, s^{i-1}, i)$, consider function
\begin{align*}
    H_i(x, s^{i-1}) = \Phi (x, s^{i-1}, i) - x
\end{align*}
$H_i(x)$ is 1-Lipschitz and monotone non-increasing by lemma~\ref{lem:app:phi_prop}. Since $ H_i(0, s^{i-1}) = \Phi (0, s^{i-1}, i) \ge 0$ and $H_i(s_K, s^{i-1}) = 0$, there exist some $z_i \in S$, such that $H_i(z_i, s^{i-1}) = 0$. This proves the existence of $z_i$.\\

\noindent Now, we show that if $x = z_i$ is positive, then there exists an optimal stopping rule that proceeds to open $b_i$. 
Fix any $i$ such that $z_i > 0$. Let $\wt{\tau}$ be the best strategy among all strategies that open $b_i$. To show that $\wt{\tau}$ is indeed optimal, we show that 
\[
\delta = \Phi(z_i, s^{i-1}, i) -\Phi^{\tilde{\tau}}(z_i, s^{i-1}, i) = 0
\]
Assume towards contradiction that $\delta > 0$. We have
\begin{align*}
0 < \delta & = \Phi(z_i,s^{i-1}, i) - \Phi^{\wt{\tau}}(z_i,s^{i-1}, i) \\
& \le \Phi(z_i,s^{i-1}, i) - \Phi^{\tau^*(z_i-\epsilon, s^{i-1}, i)}(z_i,s^{i-1}, i) \\
&  = (\Phi(z_i,s^{i-1}, i)- \Phi(z_i - \epsilon,s^{i-1}, i)) + (\Phi(z_i - \epsilon,s^{i-1}, i) - \Phi^{\tau^*(z_i-\epsilon, s^{i-1}, i)}(z_i,s^{i-1}, i)) \\
& \le 2\epsilon
\end{align*}
where we used Lipschitzness of $\Phi$ for the last inequality, and the first inequality comes from the fact that $\tau^*(z_i-\epsilon, s^{i-1}, i)$ is a sub-optimal policy that opens $b_i$. We have $\tau^*(z_i-\epsilon, s^{i-1}, i) \geq i$ since $z_i$ is the smallest such that $H_i(z_i, s^{i-1}) = 0$, this implies that $H_i(z_i - \epsilon, s^{i-1}) = \Phi^{\tau^*(z_i-\epsilon, s^{i-1}, i)}(z_i-\epsilon, s^{i-1}, i) - (z_i-\epsilon) > 0$ meaning the optimal policy will accumulate more reward than current best, thus it has to open $b_i$. As $\epsilon \rightarrow 0$, we get a contradiction. On the other hand, $H_i(z_i, s^{i-1}) = \Phi(z_i,s^{i-1}, i) - z_i = 0$ implies that the strategy that stops at $b_{i-1}$ is also optimal. Thus, $z_i$ is indeed the value for which we are indifferent between stopping and proceeding optimally.
\end{proof}

\subsection{Payoff Table}\label{subsec:app:exp_payoff_table}

\begin{algorithm}[!ht]\caption{Expected Equivalent Reward Computation, Single Hyperbox}\label{alg:app:gw_1path}
\begin{algorithmic}[1]
 \Require Ordered set of boxes $\{b_1,\ldots, b_n\}$, probing cost $\{c_1,\ldots, c_n\}$, distributions of the random payoff of boxes
 \State Initialize $z \gets 0$
 \For{$x \in S$} \Comment{Base case: filling in $T(\cdot, \cdot, n)$}
    \For{$s \in S$}
    \State $z \gets \sum_{y\in S} (\max\{x, y\}- c_n)\cdot \Pr(R_n = y)$
    \If{$z > x$}
        \State $\Phi (x, s, n) = z$, $\mathds{1}(x, s, n) = 1$ 
    \Else
        \State $\Phi (x, s, n) = x$, $\mathds{1}(x, s, n) = 0$
    \EndIf
    \State $\frm(x, s, n) = \mathds{1}(x, s, n) \cdot R_n$, and $\frc(x, s, n) = \mathds{1}(x, s, n) \cdot c_n$
    \EndFor
\EndFor
\For{$i = n-1, \cdots, 1$} \Comment{Filling in $T(\cdot, \cdot, i)$ for all $i = n-1, \cdots, 1$}
    \For{$x \in S$}
        \For{$s \in S$}
        \State $z \gets \E \Big[\sum_{y\in S} \Big(\max\Big\{x, y, \frm(x, s_y, i+1)\Big\}- c_j - \frc(x, s_y, i+1)\Big)\cdot \Pr(R_i = y)\Big]$ where $s_y$ is the state that gives $R_i$ realization $R_i = y$
        \If{$z > x$}
        \State $\Phi (x, s, i) = z$, $\mathds{1}(x, s, i) = 1$
        \Else
        \State $\Phi (x, s, i) = x$, $\mathds{1}(x, s, i) = 0$
        \EndIf
        \State Calculate $\frm(x, s, i)$ and $\frc(x, s, i)$ as follows: with probability $\Pr(R_i = y)$, $\frm(x, s, i)$ is $\mathds{1}(x, s, i) \cdot \max\{y, \frm(x, s_y, i+1)\}$ and $\frc(x, s, i)$ is $\mathds{1}(x, s, i) \cdot (c_i + \frc(x, s_y, i+1))$
        \EndFor
    \EndFor
\EndFor
\State \textbf{Return} $\Phi(x, s, i)$ for all $x\in S$, $s\in S$ and $i\in [n]$
\end{algorithmic}
\end{algorithm}

\begin{lemma}[Efficient Computation of Payoff Table]\label{lem:app:computing_rv}
    There is an efficient algorithm that computes $\phi(x,s,i)$ for all $i$, $x$ and $s$. 
\end{lemma}

\begin{proof}
\noindent Now we give an efficient algorithm for computing generalized reservation value. In fact, we will give an algorithm that uses dynamic programming to compute $\Phi (x, s^{i-1}, i)$ for all triples $(x, s^{i-1}, i)$. Then, given the current state of the algorithm $(x, s^{i-1}, i)$, the reservation value $z_i$ for box $i$ is the smallest $x$ in the table where $\Phi (x, s^{i-1}, i) = x$.\\

\noindent Denote by $T(x, s^{i-1}, i)$ our three dimensional dynamic programming table. Each entry $T(x, s^{i-1}, i)$ will store the following information:
\begin{enumerate}
    \item Expected future reward: $\Phi (x, s^{i-1}, i)$
    \item Indicator: $\mathds{1}(x, s^{i-1}, i)$ indicating whether the optimal policy will open $b_i$ in this state
    \item The distribution of future random max reward\footnote{The randomness comes from both random stopping time $\tau^*$ and correlated random variables $R_i$'s,}: $\frm(x, s^{i-1}, i):=\max_{j=i}^{\tau^\star(x, s^{i-1}, i)} R_j$ where $R_j$'s are the correlated random rewards for miniboxes that are yet to be opened given that the algorithm is at state $(x, s^{i-1}, i)$. 
    \item The distribution of future random cost\footnote{The randomness comes from $\tau^*$ being a random stopping time.}: $\frc(x, s^{i-1}, i):=\sum_{j=i}^{\tau^\star(x, s^{i-1}, i)} c_j$
\end{enumerate}

\noindent Algorithm \ref{alg:app:gw_1path} describes how to fill in the dynamic programming table. Since all random variables that appear in Algorithm \ref{alg:app:gw_1path} has finite support with size bounded by $\poly (K,n)$, and any $\max$ operation for random variables only has three or less arguments, it follows that algorithm \ref{alg:app:gw_1path} takes polynomial time and space.
\end{proof}

\section{More Details from Multiple Lines}\label{sec:app:multi_line}

\subsection{Equivalent Box}\label{subsec:app:equ_box_hyperbox}

\begin{lemma}[Equivalent Single Box for Hyperbox]\label{lem:app:equ_box}
For a stopping time $\tau$ and a hyperbox $\hyperb := \{b_1, \ldots, b_n\}$, there exists a box $\hat{b}$ with random cost (Def.~\ref{def:pandora_box_random_cost}) such that following $\tau$ over $\hyperb$ has the same utility distribution as $\hat{b}$.
\end{lemma}

\begin{proof}
We define $R = (R_1, \ldots, R_n)$ as a realization of the joint reward distribution in hyperbox $\hyperb$. For each distinct realization of $R$, the stopping time uniquely determines the payoff and cumulative cost of the hyperbox.

To construct the reward and cost distribution of equivalent single box, we define a coupling between the realizations of $R$ and both its reward and cost. When the joint reward is $R$, we assign the single box's reward as $\max_{i=1}^{\tau} R_i$ and the single box's cost as $\sum_{i=1}^{\tau} c_i$, where $\tau$ is the stopping time. The probability of each outcome matches the probability of $R$ under the original hyperbox’s joint distribution.

\end{proof}

From our construction of this box, we get the following immediate lemma that the generalized reservation value remains well-defined for the box with random cost, even when the boxes inside have random costs.

\begin{lemma}[Extending GRV to hyperboxes with random cost]\label{lem:app:GRV_multi_random}
Given a Markovian hyperbox $\hyperb := \{b_1, \ldots, b_n\}$, where each box now have stochastic cost that is correlated with the reward distribution, the GRV of each box with random cost is well-defined, and can be calculated in polynomial time. 

Moreover, if the GRV $\hat{\sigma}$ of the boxes alone, i.e., the generalized reservation value if there is only one box with random cost, satisfies: 
\begin{align*}
    \hat{\sigma} (b_1) \geq \hat{\sigma} (b_2 | R_1) \geq \ldots, \geq \hat{\sigma} (b_n | R_{n-1})
\end{align*}
for any realized reward $R_1, \ldots, R_n$, then the GRV of $b_i$ in the hyperbox only depends on $b_i$ itself. 
\end{lemma}

\begin{proof}
    Notice that $\phi$ and $H$ function remains well-defined and the properties of those remains valid, then we get that GRV is still well defined for hyperboxes with boxes of random cost. 

    The second part of the lemma follows from lem.~\ref{lem:app:properties_GRV}.
\end{proof}

From Thm.~\ref{thm:ML_MPB}, we already showed that the GRV is optimal for multi-line cases, this allows us to show how to use one random box to mimic the payoff distribution of Markovian Pandora's box with multi line constraint.

\begin{lemma}[Equivalent Boxes for Multi Lines]\label{lem:app:equ_box_multi_line}
    Given an instance of Markovian Pandora's box with multi-line precedence graph, there exist a box with random cost, such that the reward and cost of the random box is the same as the distribution of maximum reward and culmulative cost of the optimal probing strategy over the Pandora's box instance. 
    
    In addition, the GRV of the box is the maximum GRV of the available boxes when no boxes are openend.
\end{lemma}

\begin{proof}
     Then, notice that for each realization $s$ of the joint distribution of the remaining boxes, the selected boxes $\mathcal{O}$ are uniquely determined. This implies that we can generalize our construction of the box $b$ with random cost to multi-lines. With probability $\Pr[s]$, 

     \begin{align*}
         c_b &= \sum_{i \in \mathcal{O}_{s}} c_i \\
         R_b & = \max_{i \in \mathcal{O}_{s}} R_i
     \end{align*}
     where $\mathcal{O}_{s}$ denoted opened boxes under realization $s$. 

     Notice that the hyperbox with maximum GRV will be ranked first in the optimal strategy, by Lem.~\ref{lem:app:properties_GRV}, we showed that the GRV of box $b$ should equal the maximum GRV of the available boxes when no boxes are openend.
    
\end{proof}

 Then, notice that for each realization $s$ of the joint distribution of the remaining boxes, the selected boxes $\mathcal{O}$ are uniquely determined. This implies that we can generalize our construction of the box with random cost to multi-lines: 

\subsection{Proof Details from Multiple Lines}\label{subsec:app:proof_detail}

\begin{lemma}[Probing Equivalent Boxes]\label{lem:app:three_box_lemma}
Given three Pandora's boxes $A, B, C$ with random reward and random cost, with the following property:
\squishlist
    \item For each hyperbox, the reward and the cost are correlated. 
    \item The reward and cost of $B$ is independent of the reward and cost of $A$.
    \item The reward and cost (hence payoff) of $C$ depends on both $A$ and $B$ in a Markovian fashion.
    \item The reservation value $\sigma (A) > \sigma (B)> \sigma(C)$, given any realizations of $A$ and $B$ \footnote{Notice that here $\sigma(A)$ is not a random variable, but $\sigma(C)$ is a random variable depends on $A$ and $B$. }, i.e., for any possible value of $x$ of $A$ and possible value of $y$ of $B$: 
    \begin{align*}
        \sigma (A) \geq \sigma (B) \geq [\sigma(C) | R_{A} = x, R_{B} = y]
    \end{align*}
    \item We have a precedence constraint that $A$ and $B$ must be probed before $C$. 
\squishend
then conditioned on any competitive reward $X$, the optimal probing strategy is to probe these boxes in decreasing order of their generalized reservation value, i.e., probe $A$ then $B$ then $C$. 
    
\end{lemma}

\begin{proof}
It's sufficient to compare two strategies: 1)\ $D_1: B \to A \to C$, and 2)\ $D_2: A \to B \to C$. Notice that if $X > \sigma (A)$, then it's optimal to not probe any box, then the ordering of the boxes doesn't matter. WLOG, we may assume $X < \sigma (A)$. 

We first write down the expected reward according to ordering strategy $D_1$, if we use notation as in Table.~\ref{tab:reward_D_1}, we have that this reward is equivalent to: 
\begin{align*}
    &-\E[c_B] + \E[R_B | \pi_B] \Pr[\pi_B] \\
    &+\lambda_B [- \E[c_A] + \E[R_A | \pi_A \pi_A + \E[\max\{R_A, R_B, y\} | \lambda_B \cap \lambda_A] \lambda_A] + \E[\max\{R_B, y | \rho_A, \lambda_B\} \rho_A]] \\
    &+ \rho_B [ -\E[c_A] + \E[R_A | \pi_A] \pi_A + \E[\max\{y,R_A | \lambda_A\} \lambda_A + \E[\phi_C (\{R_A, R_B, X\}|\rho_A \cap \rho_B)]]] 
\end{align*}
Here, we abuse the notation $ \rho, \pi, \lambda $ to represent both events and their probabilities, with their meanings remaining unambiguous in the mathematical expressions. In addition, we use $E \cap F$ to denote the event that event $E$ and $F$ both happen. 

Similarly, using the notations in Table~\ref{tab:reward_D_2}, we have that the expected reward according to strategy $D_2$ is: 
\begin{align*}
    &-E[c_A] + \pi_A \E[R_A | \pi_A] + \lambda_A \E[\max\{X, R_A\} | \lambda_A] \\
    & + \rho_A [-\E[c_B] + \pi_B \E[R_B | \pi_B] + \lambda_B \E[\max\{ X, R_B\} | \lambda_B \cap \rho_A] + \rho_B \phi_C(\max\{X, R_B, R_A\}|\rho_A \cap \rho_B)]
\end{align*}

Notice that the last term of both payoffs can be cancelled out. Also notice that the reservation value for box $A$ and $B$ satisfies:
\begin{align*}
    \E[(R_B - \sigma_B)_{+} - c_B] = 0
\end{align*}
which simplifies to:
\begin{align*}
    \E[c_B] = \Pr[R_B \geq \sigma_B] \E[R_B | R_B \geq \sigma_B]
\end{align*}
Plugging in the appropraite values of $ \rho, \pi, \lambda $, we have:
\begin{align*}
    \E[c_A] &= \pi_A \E[R_A | \pi_A] - \pi_A \sigma_A \\
    \E[c_B] &= \pi_B \E[R_B | \pi_B] + \lambda_B \E[R_B | \lambda_B] - (\pi_B + \lambda_B) \sigma_B 
\end{align*}
Now, plugging the value of the expected cost and after simplification, we have that:
\begin{align*}
    & \E[\util (D_2) - \util (D_1)]\\
    &=  \pi_B \pi_A (\sigma_A - \sigma_B) \\
    &+ \pi_A \lambda_B [\E[R_B | \lambda_B] - \sigma_B] + \lambda_A \pi_B [\E[\max\{X, R_A|\lambda_A\} - \sigma_B]]]  \\
    & + \lambda_B \lambda_A [-\E[\max\{R_A, R_B. X\}| \lambda_B \cap \lambda_A] -\sigma_B \\
    & + \E[R_B | \lambda_B] + \E[\max\{y,R_A\} | \lambda_A]] \\
    & > \lambda_B \lambda_A [-\E[\max\{R_A, R_B. X\}| \lambda_B \cap \lambda_A] -\sigma_B + \E[R_B | \lambda_B] + \E[\max\{y,R_A\} | \lambda_A]]
\end{align*}
where the last inequality follows by the property that $\E[A | A \geq X] > X$. 

Finally, we show that the last term is positive. Notice that:
\begin{align*}
&\E[\max\{R_A, R_B,X\} | \lambda_B \cap \lambda_A] \\
& = \sigma_B + \E[\max \{ \max\{R_A, X\} - \sigma_B, R_B -\sigma_B \} | \lambda_B \cap \lambda_A] \\
& \leq \sigma_B + \E[\max\{R_A, X\} - \sigma_B + R_B - \sigma_B | \lambda_B \cap \lambda_A] \\
& = \E[R_B | \lambda_B ] + \E[\max \{X, R_A\} | \lambda_A]- \sigma_B
\end{align*}
where the last equality follows from the independence of box $A$ and $B$. Aggregating all of the above we have: 
\begin{align*}
    \E[\util (D_2) - \util (D_1)] > 0.
\end{align*}
\end{proof}

\begin{table*}[h]
\centering
\begin{tabular}{l|ccc}
\toprule
 & $R_A \geq \sigma_A$ & $\sigma_A \in (\sigma_B, \sigma_A)$ & $R_A \leq  \sigma_B $ \\
 & ${\color{blue} \pi_A}$ & ${\color{blue} \lambda_A}$ & ${\color{blue} \rho_A}$ \\
 \midrule
$R_B \geq \sigma_A$ & $\E[R_B | \pi_B]$ & $\E[R_B | \pi_B]$ & $\E[R_B | \pi_B]$ \\
${\color{blue} \pi_B}$, stop at $B$. & -$\E[c_B | \pi_B] $ & -$\E[c_B | \pi_B] $ & -$\E[c_B | \pi_B] $ \\
\midrule
$R_B \in (\sigma_B, \sigma_A)$ & $\E[R_A | \pi_A] $ & $\E[\max\{ R_A, R_B, X | \lambda_A \cap \lambda_B\}]$ & $\E[\max\{R_B, y\} | \lambda_B]$ \\
 ${\color{blue} \lambda_B}$, open $A$.  & -$\E[c_B |  \lambda_B] - \E[c_A | \pi_A]$ & $-\E[c_B | \lambda_B] - \E[c_A | \lambda_A]$ & $-\E[c_B | \lambda_B] - \E[c_A | \rho_A]$ \\
\midrule
$R_B \leq \rho_B$ & $\E[R_A | \rho_B]$ & $\E[ \max\{X, R_A \} | \lambda_A] $ & $\E[\phi_C (\{R_A, R_B, X\}|\rho_A \cap \rho_B)]$ \\
${\color{blue} \rho_B}$ & $-\E[c_B | \rho_B] - \E[c_A | \pi_A]$ & $-\E[c_B | \rho_B] - \E[c_A | \lambda_A]$ & $-\E[c_B | \rho_B] - \E[c_A | \rho_A]$ \\
\bottomrule
\end{tabular}
\caption{Table for Expected Payoff according to Strategy $D_1$. This table presents the expected total payoff for every possible joint distribution of box $A$ and $B$. In addition, $\rho_A = 1 - \pi_A - \lambda_A$.}
\label{tab:reward_D_1}
\end{table*}

\begin{table*}[h]
\centering
\begin{tabular}{l|ccc}
\toprule
 & $R_A \geq \sigma_A$ & $\sigma_A \in (\sigma_B, \sigma_A)$ & $R_A \leq  \sigma_B $ \\
 & ${\color{blue} \pi_A}$, stop at $A$ & ${\color{blue} \lambda_A}$, stop at $A$ & ${\color{blue} \rho_A}$ \\
 \midrule
$R_B \geq \sigma_A$ & $\E[R_A | \pi_A]$ &  $\E[\max\{X, R_A\} | \lambda_A]$ & $\E[R_B | \pi_B]$ \\
${\color{blue} \pi_B}$ & $-\E[c_A | \pi_A]$ & $-\E[c_A | \lambda_A]$ & $-\E[c_A | \rho_A] - \E[c_B | \pi_B]$ \\
\midrule
$R_B \in (\sigma_B, \sigma_A)$ & $\E[R_A | \pi_A]$  & $\E[\max\{X, R_A\} | \lambda_A]$ & $\E[\max\{X, R_B\} | \lambda_B \cap \rho_A]$ \\
 ${\color{blue} \lambda_B}$ & $-\E[c_A | \pi_A]$ & $-\E[c_A |\lambda_A]$ &  $-\E[c_B |\lambda_B] - \E[c_A | \rho_A]$ \\
\midrule
$R_B \leq \rho_B$ & $\E[R_A | \pi_A]$ & $\E[\max \{X, R_A\} | \lambda_A]$ & $\phi_C(\max\{y, R_B, R_A\} | \rho_A \cap \rho_B) $  \\
${\color{blue} \rho_B}$ & $-\E[c_A | \pi_A]$ & $-\E[c_A | \lambda_A]$ & $-\E[c_B |\rho_B] - \E[c_A | \rho_A]$  \\
\bottomrule
\end{tabular}
\caption{Table for Expected Payoff according to Strategy $D_2$. This table presents the expected total payoff for every possible joint distribution of box $A$ and $B$. In addition, $\rho_B = 1 - \pi_B - \lambda_B$.}
\label{tab:reward_D_2}
\end{table*}

\begin{theorem}[Polynomial Time Implementation of Generalized Reservation Value ]\label{thm:app:GMR_path}
GRV for multi-line setting can be implemented in polynomial time and space. 
\end{theorem}

\begin{proof}

We first analyze the space. One could use the same reservation value lookup table as in the single line setting, which takes polynomial space (Lem.~\ref{lem:computing_rv}). Notice that the optimal strategy, characterized by GRV, only need to enter and leave at most a finite number of hyperboxes, since the total number of boxes is finite. For each of the hyperbox visit, we only need to store the GRV of the competitive boxes and the GRV when the strategy last enter this hyperbox, which takes polynomial space. 

Next, we analyze the space complexity. The lookup for the generalized reservation value (GRV) can still be performed using binary search, as in Theorem~\ref{thm:GMR_path}, which runs in polynomial time. Since the GRV is computed at most once per box, the overall runtime remains polynomial.

\end{proof}

\section{More Details from Forest Setting}\label{sec:app:omit_forest}

\subsection{Preliminaries on Graph}\label{subsec:app:pre_graph}

\begin{definition}[Component]\label{def:app:component}
Given an undirected graph $G = (V, E)$, a \textbf{component} of $G$ is a maximal connected subgraph $C = (V_C, E_C)$ such that:
\begin{itemize}
    \item $C$ is connected: There exists a path between any two vertices in $V_C$.
    \item $C$ is maximal: No additional vertex $v \in V \setminus V_C$ can be included without losing connectivity.
\end{itemize}
A graph is said to be \textbf{connected} if it consists of a single component.
\end{definition}

\begin{definition}[Induced Subgraph]\label{def:app:induced_subgraph}
Given a graph $G = (V, E)$ and a subset of vertices $V' \subseteq V$, the \emph{induced subgraph} $G[V']$ is the graph $(V', E')$ where:
\begin{align*}
E' = \{ (u, v) \in E \mid u, v \in V' \}
\end{align*}
That is, $G[V']$ contains all edges from $G$ whose endpoints are both in $V'$.
\end{definition}

\subsection{Generalized Reservation Value for Forest Setting}\label{subsec:app:GRV_forest}

We begin by presenting Fig.~\ref{fig:app:subtree_intuition}, which explains the contraction step of our algorithm. Starting from a minimal tree, condition on any values of $A$, we consider a multi-line Pandora's box consists of boxes other than $A$, then we contract them to one single box with random cost.

\begin{figure}[htbp]
    \centering
    \subfigure[Original Subtree]{%
        \includegraphics[width=0.25\textwidth]{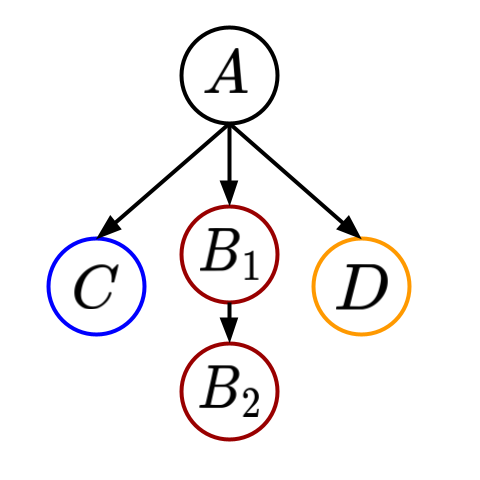}%
    }%
    \hfill
    \subfigure[Reduction to Multi-Lines]{%
        \includegraphics[width=0.25\textwidth]{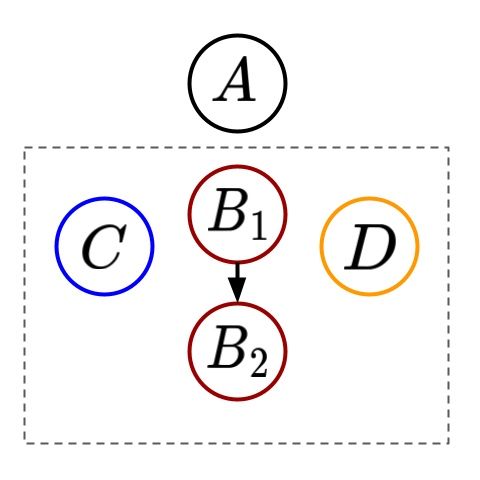}%
    }%
    \hfill
    \subfigure[Reordering by GRV(Adaptive)]{%
        \includegraphics[width=0.4\textwidth]{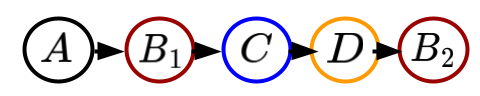}%
    }%
    \caption{Reduction From Tree to Multi-Line Setting}
    \label{fig:app:subtree_intuition}
\end{figure}

Next, we define the GRV for the forest setting. 

\begin{definition}[GRV, Forest]\label{def:app:GRV_forest}
    Given a Markovian Pandora's box with precedence graph $G = (\B, E )$ as a forest, and given the opened boxes as $\mathcal{B}_{o}$ and the information set on the realized reward $\mathcal{I}_{o} := \{ b_i = R_i, $\text{~for every~} $i \in \mathcal{B}_{o}\}$. 
    
    For any unopened box $b_i$, the GRV of $b_i$ can be derived by applying Alg.~\ref{alg:pd_forest} over the component of the induced subgraph of $G[\B \setminus \mathcal{B}_{o}]$ conditioned on the current information set $\mathcal{I}_{o}$.
\end{definition}

\begin{proof}
When calculating the GRV for each of the boxes inside our algorithm, the box is either the root of a minimal tree or the current graph consists of multiple lines. It suffices to show that this GRV is well-defined for either cases.

From the results from the multi-line cases, we get that if the graph is lines or vertices, then GRV of any box is well defined. Furthermore, the entire graph could be contracted by one box with random cost (Lem.~\ref{lem:app:equ_box_multi_line}). 

We first define the GRV of the root of a minimal tree. Before calculating the GRV, we first define how to calculate the equivalent reward:  We first contract the induced subgraph consists of vertices other than the root node as a single box $\hat{b}$ with random cost, then use the following equation to calculate the equivalent reward for $r$ given the current reward $x$:
  \begin{align*}
     \Phi(x,r) = & \max\{x, \E[\max\{R_r, x\} - c_r], \\
     & \quad \E[\max\{R_r, x, R_{\hat{v}}\} - c_r - c_{\hat{v}}]\}
 \end{align*}
The terms in the max operator correspond to the utility of not opening $r$, opening only $r$, and opening $r$ while optimally exploring the remaining nodes, respectively. Consequently, we extend the definition of GRV to the root vertex $r$ as the smallest $x$ satisfying $\Phi(x, r) = x$. The GRV is well defined since this minimal tree can be treated equally as a hyperbox with the first box is $r$, and the second box is $\hat{v}$, then by Lem.~\ref{lem:app:GRV_multi_random}, the GRV is well defined. 

\end{proof}

\begin{lemma}[Time and Space Complexity of Updating GRV]\label{lem:app:update_GRV_forest}
    The algorithm~\ref{alg:pd_forest} can be implemented in polynomial time and space. 
\end{lemma}

\begin{proof}
\textbf{Space Complexity}. During the execution of our algorithm, each minimal tree can be found using \emph{BFS} or \emph{DFS}, both of which run in polynomial time. The maximum number of times we need to find a minimal tree corresponds to the number of boxes in the forest, ensuring that this step remains polynomial overall.

The worst-case scenario for computing any GRV occurs when the given node is the \textit{root of the minimal tree} containing all boxes. For non-root nodes, $\phi$ can be computed in polynomial time using the multi-line case runtime. For the root $r$, we apply \textit{binary search} similarly as in previous analysis, which also runs in polynomial time.

Thus, the entire algorithm is implemented in polynomial time.

\textbf{Space Complexity}. We only need to store the $\phi$ table for all possible states of each box, along with the entire graph as required by our algorithm. By similar arguments from previous sections, storing the $\phi$ table requires at most polynomial space. The remaining operations also fit within polynomial space constraints.

\end{proof}

Since probing according to latest GRV will at most probe all the boxes, Alg.~\ref{alg:pd_forest} are invoked for at most the number of total boxes in the forest. Thus the overall running time is still polynomial.

\begin{theorem}[Optimality of GRV, Forest]\label{thm:app:GRV_opt_forest}
The GRV defined as in Def.~\ref{def:app:GRV_forest} is optimal for probing the forest. 
\end{theorem}

\begin{proof}
Notice that each contraction step preserves the \textit{equivalent reward table} and the \textit{payoff distribution} of probing the minimal tree. Since the algorithm updates $\phi$ for the root of the minimal box using the strategy that maximizes expected reward, the $\phi$ values—and thus the GRV—of its parent boxes remain unchanged under contraction.  

This allows us to use the contracted graph $\hat{G}$, which consists only of \textit{multiple lines and single boxes}, to compute the $\phi$ table and the GRV $\sigma$ for each available box. Applying Thm.~\ref{thm:ML_MPB} to $\hat{G}$ then establishes the optimality of GRV.

\end{proof}

\section{More Details from Static Transition}\label{sec:app:omit_rapid_mix}

We formally define the distribution of a hyperbox's payoff.

\begin{definition}[Payoff Distribution of Optimal Stopping]\label{def:stat_dist}
Given an adaptive stopping time $\tau$ (i.e., decides whether to stop based on reward realizations) of a hyperbox $\hyperb := \{b_1, \ldots, l_m\}$ as $\tau$, we call the distribution of the payoff as the distribution of the following quantity:

\begin{align*}
\util_\tau(\hyperb) = \max_{j = i}^{\tau} R_j - \sum_{j=i}^{\tau} c_j
\end{align*}
We also denote $\util_\tau(\hyperb) = \max_{\tau} \max_{j = i}^{\tau} R_j - \sum_{j=i}^{\tau} c_j$, and denote $\tau^*$ as the optimal (adaptive) stopping time that obtains this maximum. 
\end{definition}

Notice that $\D_i$, the reward distribution of the first box, corresponds to the initial probability distribution of the Markov Chain, and can be presented as a probability (row) vector $\pi_0$. Now we are ready to present the reward distribution of the $i$-th probe:

\begin{definition}[Multi-Probe Reward]\label{def:multi_probe_reward}
The $t$-th transition reward $p_{i,j}^{(t)} = \Pr[R_{t+1} = v_j | R_1 = i]$, and it satisfies $p_{i,v}^{(t)} = (\trans^{t})_{i,j}$. The reward distribution of $t$-th probe is: $\pi_0 \trans^{t-1}$. 
\end{definition}

\subsection{Preliminaries from Markov Chain}\label{subsec:app:MC_preliminary}


We first introduce a distance metric measuring the distance between distributions over the same domain.

\begin{definition}[TV-distance]\label{def:TV_dis}
Given a probability distribution $\mu$ and $\nu$ on same domain $\Omega$, the total variation (TV) distance between them is:
\begin{align*}
    \dtv(\mu, \nu) = \max_{A \subseteq \Omega} |\mu(A) - \nu(A)|. 
\end{align*}
Alternatively, $\dtv (\mu, \nu) = \min_{(X,Y)} \Pr (X \neq Y) $, where the minimum is over all coupling (joint distribution) of the distribution $\mu, \nu$, such that $X \sim \mu$ and $Y \sim \nu$.  
\end{definition}

Next, we present the definition of Markov Chain. Notice that for every directed line as a subgraph of the precedence graph, the reward along the line follows a Markov Chain. 

\begin{definition}[Markov Chain]\label{def:app:MC}
A \textit{Markov chain} is a discrete-time stochastic process $\{X_n\}_{n \geq 0}$ on a state space $S$ satisfying the \textbf{Markov property}:
\[
P(X_{n+1} = j \mid X_n = i, X_{n-1} = i_{n-1}, \dots, X_0 = i_0) = P(X_{n+1} = j \mid X_n = i).
\]
The transition probabilities form a matrix $P$, where each entry $P_{ij}$ represents the probability of transitioning from state $i$ to state $j$:
\[
P_{ij} = P(X_{n+1} = j \mid X_n = i), \quad \sum_{j \in S} P_{ij} = 1 \quad \forall i \in S.
\]
\end{definition}

A Markov chain of certain property would have analytical property. We present such property below. In the main text, we present a similar definition that relies solely on the transition matrix, independent of the concepts of the Markov chain.

\begin{assumption}[Properties of Markov Chain]\label{ass:app:MC_property}
A Markov chain with transition matrix $P$ and state space $\{v_1, \ldots, v_k\}$ satisfies: 

\begin{itemize}
    \item \textbf{Irreducibility}, if for all $v_i, v_j \in \{v_1, \ldots, v_k\}$, there exists an integer $n \geq 1$ such that:
    \begin{align*}
    P^n(v_i, v_j) > 0,
    \end{align*}
    where $P^n(v_i, v_j)$ denotes the $(i,j)$-entry of $P^n$, representing the probability of transitioning from state $v_i$ to state $v_j$ in $n$ steps.

    \item \textbf{Aperiodic}, if for every state $v_i$, the greatest common divisor (gcd) of the set $\{n \geq 1 : P^n(v_i, v_i) > 0\}$ is 1. Formally:
    \begin{align*}
    d = \text{gcd}\{n \geq 1 : P^n(v_i, v_i) > 0\},
    \end{align*}
    and the chain is aperiodic if $d = 1$ for all $v_i \in \{v_1, \ldots, v_k\}$.
\end{itemize}
\end{assumption}

\begin{definition}[Stationary Distribution]\label{def:MC_stable}
We say a probability column vector $\pi$ is a \emph{stationary distribution} for a Markov chain with transition matrix $\trans$, if $\pi \trans = \pi$.
\end{definition}

\begin{lemma}[Existence of Unique Stable Distribution]\label{lem:exist_stable_dist}
If a finite-state Markov chain is \textit{irreducible} and contains at least one \textit{aperiodic} state, then the chain has a \textbf{unique stationary distribution} $\pi$. Furthermore, the chain converges to $\pi$ regardless of the initial state.
\end{lemma}

\begin{lemma}[Convergence to Stable Distribution~\cite{lpw17}]\label{lem:app:conv_reward}
    Given an aperiodic and irreducible Markov chain $\{R_t\}_{t \in [\mathbb{N}]}$ with state space $\V$, transition matrix $\trans$. There exist a unique stationary distribution $\pi$, and there exist $C_{\trans} > 0, \alpha_{\trans} \in (0,1)$, such that:

    \begin{align*}
    \max_{v_i \in \V} \| {\trans}_{i, :}^t - \pi \|_\text{TV} \leq C_{\trans} \alpha_\trans^t
    \end{align*}
\end{lemma}

This lemma implies that given any starting state $v_i$, the reward distribution if the Markov Chain converges to a stable distribution in its TV distance. 

Next, we define the mixing time of this Markov Chain:

\begin{definition}[Mixing Time of Markov Chain~\cite{lpw17}]\label{def:app:mix_time_MC}
Given a Markov chain with unique stable distribution $\pi$ and transition matrix $P$, the mixing time is defined as the time for a Markov Chain to reach total variation distance of within a given parameter $\eps$ of $\pi$, i.e., 

\begin{align*}
    \tmix (\eps) = \min \{t : \max_{v_j, j \in [k]} \| P_{x, \cdot}^{t} - \pi \|_\text{TV} \leq \eps\}
\end{align*}

Moreover, we denote $\tmix = \tmix(1/4)$ and we have $\tmix (\eps) \leq \lceil \log_2 \nicefrac{1}{\eps} \rceil \tmix $.
    
\end{definition}


    


\begin{lemma}[Convergence of (Unconditional) Reward]\label{lem:app:convege_line_stable}
Given a Markov Pandora's box with static transition, and that the transition matrix of every maximal component of the graph is irreducible and aperiodic. Given any directed line subgraph $\hyperb := \{b_1,\ldots, b_n\}$ of the original graph with transition matrix $P$, there exists a unique stable distribution $\pi$, for which the (unconditional) reward distribution of $b_n$ will converge to, once $n \to \infty$. Equivalently, there exist $C_{\trans} > 0, \alpha_{\trans} \in (0,1)$, such that for every $n \in \mathbb{N}$, the reward distribution $\pi_n$ of $b_n$ satisfies:
    \begin{align*}
    \| \pi_n - \pi \|_\text{TV} \leq C_{P} \alpha_P^n
    \end{align*}
    
\end{lemma}

\begin{proof}
    Then by Lem.~\ref{lem:app:conv_reward} and Lem.~\ref{def:app:mix_time_MC}, we have that:
    \begin{align*}
        \max_{i \in [k]} \| {P}_{i, :}^n - \pi \|_\text{TV} \leq \hat{C}_{P} \alpha_P^n
    \end{align*}
    Given any $\hyperb$, we denote the reward distribution of its first box as $\hat{\pi}$, then we have the distribution of $n$-th box in line satisfies:
    \begin{align*}
        \| \pi^n -\pi \|_\text{TV} = \| \hat{\pi} P^{n-1} -\pi \|_\text{TV}
        \leq & \sum_{i \in [k]} \hat{\pi}_{i} \max_{i \in [k]} \| {P}_{i, :}^{n-1} - \pi \|_\text{TV}
         \leq  \max_{i \in [k]} \| {P}_{i, :}^{n-1} - \pi \|_\text{TV}  \leq \hat{C}_{P} \alpha_P^{n-1}
    \end{align*}
    Letting $C_P := \hat{C}_p / \alpha_P$, we prove that the statement in the lemma is correct.
\end{proof}

\subsection{Multi Line Setting with Static Cost}\label{subsec:app:rapid_mix_MC_static_cost}

\begin{lemma}[Arithmetic Mean–Geometric Mean (AM-GM) Inequality]\label{lem:app:AM_GM}
For any $a_1, a_2, \dots, a_n \in \R_{+}$, the arithmetic mean is greater than or equal to the geometric mean:
\begin{align*}
\frac{a_1 + a_2 + \cdots + a_n}{n} \geq \sqrt[n]{a_1 a_2 \cdots a_n},
\end{align*}
with equality if and only if $a_1 = a_2 = \cdots = a_n$.

\end{lemma}

\begin{lemma}[Banach fixed-point theorem]\label{lem:app:Banach_fix_point}
Let $(X, d)$ be a non-empty complete metric space, and let $f: X \to X$ be a contraction mapping, meaning there exists a constant $q \in [0, 1)$ such that:
\begin{align*}
d(f(x), f(y)) \leq q \cdot d(x, y), \quad \forall x, y \in X.
\end{align*}
Then:
\begin{enumerate}
    \item $f$ has a unique fixed point $x^* \in X$, such that $f(x^*) = x^*$.
    \item For any initial point $x_0 \in X$, the sequence $\{x_n\}$ defined by $x_{n+1} = f(x_n)$ converges to $x^*$.
    \item The rate of convergence is at least linear, with:
    \begin{align*}
    d(x_n, x^*) \leq \frac{q^n}{1 - q} d(x_0, f(x_0)).
    \end{align*}
\end{enumerate}
\end{lemma}

\begin{theorem}[MC with static cost]\label{thm:app:MC_constant_cost}
Under Ass.~\ref{ass:MC_cost} and each probe cost an additional cost of constant $c$, 
\begin{itemize}
    \item The optimal strategy would continue probe the box until $v_k$ is realized, under the following (sufficient) condition: For any $i \in [k-1]$, $p_{i, k} (v_k - v_{k-1}) - c > 0$.
    \item Given that for any $i \in [k-1]$, $p_{i,k}>0$. The equivalent reward $ \phi(y, x) $ takes in put $ y $ as the current maximum reward and $ x $ as the current reward. The optimal strategy is to continue if $ \phi(y, x) > y $, and stop if $ \phi(y, x) \leq y $.
\end{itemize}

\end{theorem}

\begin{proof}
    Note that in this scenario, the equivalent payoff only depends on two factors, the current max reward $x$ and current state $v_i$. Then we could solve the equivalent payoff (i.e., expected future reward/payoff $\phi(\cdot, \cdot)$) of each state $(y, v_i)$, where $y \geq v_i \in \V$ as follows:
    \begin{itemize}
        \item The expected reward of not opening the next box is $y$.
        \item The expected reward of opening the next box is $-c + \sum_{j \in [k]} p_{i,j} \phi (\max \{y, v_j\}, v_j )$.
        \item By definition, we have the following Bellman equation:
        \begin{align*}
            \phi (v_l,v_i) & = \max \{ v_l, -c + \sum_{j \in [k]} p_{i,j} \phi (\max \{y, v_j\}, v_j) \} \\
            & = \max \{ v_l, -c + \sum_{l \in[i]} p_{j,l} \phi(i,l) + \sum_{l=i+1}^{k-1} p_{j,l} \phi(l,l) + p_{j,k} v_k \}
        \end{align*}
        \item Next, we show that the $\phi(y, v_i)$ can be solved via fixed point iteration, and there exists a unique $\phi$ such that all equations are satisfied. We denote vector $\Phi \in [0,v_k]^{k(k-1)}$, and function $f$ that takes $\Phi$ as input, and output a $k(k-1)$ dimensional vector, for any $i \in [k-1], j \in [k]$: 
        
        \begin{align*}
            f(\Phi)_{i\cdot k+j} = \max \{ x_i, -c + \sum_{l \in[i]} p_{j,l} \Phi_{i \cdot k + l} + \sum_{l=i+1}^{k-1} p_{j,l} \Phi_{l \cdot k + l} + p_{j,k} v_k   \}
        \end{align*}
        Notice that the $\Phi$ is not $k^2$ dimensional as $\phi(v_k, y) = v_k$ for all $y < v_k$, so these function values are not determined through fixed-point iteration. 
        
        Next we show that when applying Chebyshev distance $d$, mapping $f$ is a contraction mapping:

        \begin{align*}
            & \max_{i,j} d ( f (x)_{ik+j}, f(y)_{ik+j} ) \\
            \leq &  \max_{i,j} |\sum_{l \in[i]} p_{j,l} x_{i \cdot k + l} + \sum_{l=i+1}^{k-1} p_{j,l} x_{l \cdot k + l} - \sum_{l \in[i]} p_{j,l} y_{i \cdot k + l} + \sum_{l=i+1}^{k-1} p_{j,l} y_{l \cdot k + l}|  \\
            \leq & \sum_{l \in [k-1] } p_{j,l} d (x,y) < d (x,y )
        \end{align*}
        where the first inequality follows from definition, and the last inequality follows from $p_{j,k} >0$. 

        Thus, from Banach fixed-point theorem(Lem.~\ref{lem:app:Banach_fix_point}), there exist a unique fixed point $\Phi^*$, such that starting from any point $\Phi_0$, and let $\Phi_t = f (\Phi_{t-1})$, $\lim_{t \to \infty} \Phi_t \to \Phi^*$, and $d (\Phi^*, \Phi^t) \leq q d (\Phi^*, \Phi^{t-1})$, for $q = \max_i (1 - p_{i,k})$. Notice that it's easy to check that we take the max correctly among the original bellman's equation, our proof is complete. 
        
        \item     We next show that given $i \in [k]$, under $p_{i, k} (v_k - v_{k-1}) - c$, for any $v_i$, $\phi (v_{k-1}, v_i) > v_{k-1}$. Notice that this argument means that the generalized reservation value of any state $v_i$ is larger than $v_{k-1}$. Notice that:

    \begin{align*}
        &-c + \sum_{j \in [k]} p_{i,j} \phi (\max \{y, v_j\}, v_j) \\
        = & -c + \sum_{j \in [k-1]} p_{i,j} \phi(v_{k-1}, v_j) + p_{i, k} v_k \\
        \geq & -c + \sum_{j \in [k-1]} p_{i,j} [\phi(v_{k}, v_j) - v_k + v_{k-1} ] + p_{i, k} v_k \\
        \geq & -c +  v_{k-1} \sum_{j \in [k-1]} p_{i,j} + p_{i,k} v_k \\
        = & v_{k-1} +p_{i,k} (v_k - v_{k-1}) -c \\
        > & v_{k-1}
    \end{align*}
    where the third line follows from the Lipshitz property of $\phi$ in its first input (Lem.~\ref{lem:app:phi_prop}), and the other line follows from reorganization.
    \end{itemize}

\end{proof}

\subsection{Results for Markovian Pandora's Box with General Cost}\label{subsec:app:multi_line_general}

\begin{theorem}[Near Optimal Probing, Single Line]\label{thm:app:greedy_optimal}
Given a Markovian Pandora’s box with static transition, where the precedence graph is a single line $ \hyperb $ and the transition matrix $ \trans $ is irreducible, aperiodic, and associated with a stable distribution $ \pi $. There exists constant $C_{\trans}>0$ and $\alpha_{\trans} \in (0,1)$, that, given any $\delta \in (0,1)$, let $t_{\delta} = \max \{ 2 \frac{C_{\trans}}{ \pi_k(1 - \alpha_{\trans})},  \frac{\log (1 / \delta)}{\log (1 - \pi_k / 2)} \}$. Then, the expected utility from optimal probing on the subgraph $ \hyperb|t_{\delta} $, containing only the first $ t_{\delta} $ boxes of $ \hyperb $, is near-optimal, i.e.,
\begin{align*}
   \OPT (\hyperb|{t_\delta}) \geq   \OPT(\hyperb) - 2\delta v_k
\end{align*}
where $\OPT(\hyperb) = \E[\util (\hyperb)]$ is the expected utility from optimal probing hyperbox $\hyperb$.
\end{theorem}

\begin{proof}
    WLOG, we denote the optimal stopping for $\hyperb$ as $\tau^*$, and the optimal stopping for $\hyperb|_{t_\delta}$ as $\tau_{\delta}$. We construct a stopping time $\bar{\tau}$ such that $\bar{\tau} = \tau^*$ for $\tau^* \leq t_\delta$, and $\bar{\tau} = t_\delta$ for $\tau^* > t_\delta$. Following the notations in Def.~\ref{def:stat_dist}, we have:
    \begin{align*}
        & \Pr( \tau^* \leq t_\delta) \E[\util(\hyperb)|\tau^* \leq t_\delta] \\
        = & \OPT_{\hyperb} - \Pr( \tau^* > t_\delta) \E[\util(\hyperb| t_{\delta})|\tau^* > t_\delta] \\
        \geq & \OPT_{\hyperb} - \delta v_k 
    \end{align*}
    where the first equality follows from the law of total expectation, the second inequality follows from Lem.~\ref{lem:dist_max_reward}.

    Next, notice that for scenarios that $\tau^* > t_\delta$, $\tau^*$ would stop at the last box when probing $\hyperb| t_{\delta}$, we have: 
    \begin{align*}
        &\E[\util(\hyperb| t_{\delta} )] \geq \E[\util_{\bar{\tau}}(\hyperb| t_{\delta} )] \\
        \geq & \E[\util(\hyperb)|\tau^* \leq t_\delta]\Pr[\tau^* \leq t_\delta] + \E[\util_{\bar{\tau}}(\hyperb)|\tau^* > t_\delta] \Pr [\tau^* > t_\delta] \\
        \geq & \OPT_{\hyperb} - \delta v_k - \delta v_k
     \end{align*}
    where the first inequality is due to the optimality of $\tau_\delta$ on $\hyperb | t_{\delta}$, and the second inequality follows from the Lem.~\ref{lem:dist_max_reward} that $\Pr[\tau^* \leq t_\delta] \geq \Pr[R_{\max} (t_\delta) = v_k)] \geq 1-\delta$. Regarding the last inequality, notice that for $\Pr[\tau^* > t_\delta] >0$, it's necessary that $v_k - \sum_{t \in [t_\delta]} c_t >0$, otherwise $\tau^*$ is not optimal. Thus, $ \E[\util_{\bar{\tau}}(\hyperb)|\tau^* > t_\delta] \geq - \sum_{t \in [t_\delta]} c_t \geq - v_k$.

\end{proof}

\begin{theorem}[Near Optimal Probing, Multi Lines]\label{thm:app:greedy_optimal_multi_line}
Given a Markovian Pandora’s box with static transition, where the precedence graph $G$ consists of $q$ directed lines $ \hyperb_1, \ldots, \hyperb_q$. Suppose every transition matrix $P_j$ of each $\hyperb_j, j \in [q]$ are irreducible, irreducible, aperiodic, and associated with a stable distribution $ \pi(j) $. There exists constant $C>0$ and $\alpha \in (0,1)$, that, given any $\delta \in (0,1)$, let $t_{\delta} = \max \{ 2 \frac{C}{ \pi^* (1 - \alpha)},  \frac{\log (1 / \delta)}{\log (1 - \pi^* / 2)} \}$, where $\pi^* = \max_{j \in [q]} \pi (j)_k$. Then, the expected utility from optimal probing on the subgraph $ \cup_{j \in [q]} (\hyperb_j |t_{\delta}) $, containing only the first $ t_{\delta} $ boxes of each hyperbox, is near-optimal, i.e.,
\begin{align*}
   \OPT [\cup_{j \in [q]} (\hyperb_j |t_{\delta})] \geq   \OPT(G) - 2q\delta v_k
\end{align*}
where $\OPT(\hyperb) = \E[\util (\hyperb)]$ is the expected utility from optimal probing hyperbox $\hyperb$.
\end{theorem}

\begin{proof}

Notice that from Lem.~\ref{lem:app:convege_line_stable}, and Lem.~\ref{thm:app:greedy_optimal}, for every $j \in [q]$, there exist $C_j$ and $\alpha_j \in (0,1)$ such that, given any $\delta \in (0,1)$, setting $t_{\delta}(j) = \max \{ 2 \frac{C_{j}}{ \pi_k (1 - \alpha_{j})},  \frac{\log (1 / \delta)}{\log (1 - \pi_k / 2)} \}$, the probability that max reward along the first $\hyperb_j|_{t_{\delta} (j)}$ smaller than $v_k$ is $\delta$. We denote event $E$ as the event that for any $j \in [k]$, the first $t_\delta (j)$ boxes of $\hyperb_{j}$ has reward $v_k$, and $Pr[E] \geq 1- q \delta$.  

WLOG, we denote the optimal stopping for $G$ as $\tau^*$, and the optimal stopping for $\cup_{j \in [q]}\hyperb_j$ as $\tau_{\delta}$. We construct a stopping time $\bar{\tau}$ such that $\bar{\tau} $: 

    \textbf{Event $\mathbf{E1}$:} is the same as $ \tau^*$ if $\tau^*$ doesn't explore any boxes with index greater than $t_{\delta} (j)$ along $\hyperb_j$ for all $j \in [q]$;
    
    \textbf{Event $\mathbf{E2}$:} stops once $\tau^*$ first hit a box with index greater than $t_{\delta} (j)$ inside any $\hyperb_j, j \in [q]$. 

Following the notations in Def.~\ref{def:stat_dist}, we have:
\begin{align*}
            & \Pr( E1) \E[\util(G)|E1] \\
        = & \OPT_{G} - \Pr( E2) \E[\util(G)|E2]  \\
        \geq & \OPT_{G} - q \delta v_k 
\end{align*}
where the last inequality follows from the fact that event $E2$ and event $E$ are mutually exclusive. 

Then we have, 
\begin{align*}
    \E[\util (\cup_{j \in [q]}\hyperb_j|t_{\delta}(j))] & \geq \E[\util_{\bar{\tau}} (\cup_{j \in [q]}\hyperb_j|t_{\delta}(j))] \\
    & \geq \Pr[E1] \E[\util (\cup_{j \in [q]}\hyperb_j)|E_1] + \Pr[E2] \E[\util_{\bar{\tau}} (\cup_{j \in [q]}\hyperb_j|t_{\delta}(j)) | E_2] \\
    & \geq \OPT_{G} - 2q \delta v_k 
\end{align*}
The last inequality is due to the fact that 1)\ for event $E2$, the additional cost from exploring boxes excluded from $\cup_{j \in [q]}\hyperb_j$ is upper bounded by $v_k$, otherwise $\pi^*$ is not optimal, and 2)\ $\Pr[E2] \leq q \delta$.   

Now, let $C = \max_{j \in [q]} C_j$ and let $\alpha = \max_{j \in [q]} \alpha_j$ gives us the theorem statement. 
\end{proof}

\begin{theorem}[$1/2$ Approximation, Forest]\label{thm:app:forest_approx}

Given a Markovian Pandora’s box with static transition, where the precedence graph $G$ is a forest. Given any $\delta \in (0,1)$,there exists an algorithm such that it finds a best fixed line subgraph $\hyperb$ within $\Delta(G)^{\tilde{\Theta}(1)}$ time, such that given any alternative adaptive policy $\pi$:
    \begin{align*}
        \E [\max_{i \in \mathcal{O}(\pi)} R_i - \sum_{i \in \mathcal{O}(\pi)} c_i] \geq 
        1/2 \cdot  \E[\max_{i \in \mathcal{O}(\hat{\pi})} R_i] - \sum_{i \in \mathcal{O}(\hat{\pi})} c_i] -q \delta.
    \end{align*}
where $\Delta(G)$ is the degree and $q$ is the number of trees in $G$. 
\end{theorem}

\begin{proof}
From lem.~\ref{lem:forest_adaptivity}, we get that by restricting ourselves to the NA strategies, we only loose a factor of $1/2$ in approximation ratio. Given $\delta_0 \in (0,1)$, let $t_\delta$, $\alpha$ and $C$ be defined as in Thm~\ref{thm:app:greedy_optimal_multi_line}. 

Note that for each tree, the transition matrix remains the same. Therefore, we only need to search for the tree that minimizes the cumulative cost of a directed line of length $1, \ldots, t_{\delta_0}$ starting from the root node. This search can be performed using a depth-first search traversal of the graph, constrained to a maximum depth of $t_{\delta_0}$. Thus, the algorithm's running time is $\Delta(G)^{t_{\delta_0}} = \Delta(G)^{\tilde{\Theta}(1)}$.

Next, notice that searching for directed line with length at most $t_\delta$ will at most cost an additional $q \delta_0 v_k$ according to similar arguments as in Thm.~\ref{thm:app:greedy_optimal_multi_line}. Thus, let $\delta = \delta_0/ v_k$ gives us the results in theorem statement.
    
\end{proof}




\end{document}